\documentclass[11pt]{amsart}
\usepackage{amsfonts}
\usepackage{amsmath,amscd}
\usepackage{amsthm}
\usepackage{amssymb}
\usepackage{latexsym}

\usepackage[usenames,dvipsnames]{color}

\usepackage[normalem]{ulem}

\usepackage{hyperref}

\newtheorem{lem}{Lemma}[section]

\newtheorem{prop}{Proposition}[section]

\newtheorem{example}{Example}[section]
\newtheorem{defn}{Definition}[section]

\newtheorem{rem}{Remark}[section]

\numberwithin{equation}{section}
\newcommand{\beqa}{\begin{eqnarray}}
\newcommand{\eeqa}{\end{eqnarray}}

\newcommand{\nc}{\newcommand}
\newcommand{\rnc}{\renewcommand}

\nc{\cal}{\mathcal}

\nc{\goth}{\mathfrak}
\rnc{\bold}{\mathbf}
\renewcommand{\frak}{\mathfrak}
\renewcommand{\Bbb}{\mathbb}

\newcommand{\fpt}[7]{{}_4\phi_3\left[\begin{matrix} #1 , #2, #3, #4 \\
#5, #6, #7 \end{matrix}\,; q^2,q^2\right]}

\newcommand{\fpts}[6]{{}_3\phi_2\left[\begin{matrix} #1 , #2, #3 \\
#5, #6 \end{matrix}\,; q^2,q^2\right]}

\nc\kJ{\mathfrak J}

\nc\K{\mathbb K}
\nc{\Cal}{\mathcal}
\nc{\Xp}[1]{X^+(#1)}
\nc{\Xm}[1]{X^-(#1)}
\nc{\on}{\operatorname}
\nc{\ch}{\mbox{ch}}
\nc{\Z}{{\bold Z}}
\nc{\J}{{\mathcal J}}
\nc{\C}{{\bold C}}
\nc{\Q}{{\bold Q}}

\nc{\N}{{\Bbb N}}
\nc\beq{\begin{equation}}
\nc\enq{\end{equation}}
\nc\lan{\langle}
\nc\ran{\rangle}
\nc\bsl{\backslash}
\nc\mto{\mapsto}
\nc\lra{\leftrightarrow}
\nc\hra{\hookrightarrow}
\nc\sm{\smallmatrix}
\nc\esm{\endsmallmatrix}
\nc\sub{\subset}
\nc\ti{\tilde}
\nc\nl{\newline}
\nc\fra{\frac}
\nc\und{\underline}
\nc\ov{\overline}
\nc\ot{\otimes}
\nc\bbq{\bar{\bq}_l}
\nc\bcc{\thickfracwithdelims[]\thickness0}
\nc\ad{\text{\rm ad}}
\nc\Ad{\text{\rm Ad}}
\nc\Hom{\text{\rm Hom}}
\nc\End{\text{\rm End}}
\nc\Ind{\text{\rm Ind}}
\nc\Res{\text{\rm Res}}
\nc\Ker{\text{\rm Ker}}
\rnc\Im{\text{Im}}
\nc\sgn{\text{\rm sgn}}
\nc\tr{\text{\rm tr}}
\nc\Tr{\text{\rm Tr}}
\nc\supp{\text{\rm supp}}
\nc\card{\text{\rm card}}
\nc\bst{{}^\bigstar\!}
\nc\he{\heartsuit}
\nc\clu{\clubsuit}
\nc\spa{\spadesuit}
\nc\di{\diamond}
\nc\cW{\cal W}
\nc\cG{\cal G}
\nc\al{\alpha}
\nc\bet{\beta}
\nc\ga{\gamma}
\nc\de{\delta}
\nc\ep{\epsilon}
\nc\io{\iota}
\nc\om{\omega}
\nc\si{\sigma}
\rnc\th{\theta}
\nc\ka{\kappa}
\nc\la{\lambda}
\nc\ze{\zeta}

\nc\vp{\varpi}
\nc\vt{\vartheta}
\nc\vr{\varrho}

\nc\Ga{\Gamma}
\nc\De{\Delta}
\nc\Om{\Omega}
\nc\Si{\Sigma}
\nc\Th{\Theta}
\nc\La{\Lambda}

\nc\op{\overline{p}}
\nc\oE{\overline{E}}
\nc\onabla{\overline{\nabla}}
\nc\oDelta{\overline{\Delta}}

\nc\bepsilon{\overline{\epsilon}}
\nc\bak{\overline{k}}
\nc\bgamma{\overline{\gamma}}
\nc\kb{\mathfrak b}
\nc\hQ{\widehat Q}
\nc\kN{\mathfrak N}
\nc\kM{\mathfrak M}
\nc\kX{\mathfrak X}
\nc\kA{\mathfrak A}

\nc\kn{\mathfrak n}
\nc\km{\mathfrak m}

\nc\tfQ{\widehat{\mathfrak Q}}

\nc\boa{\bold a}
\nc\bob{\bold b}
\nc\boc{\bold c}
\nc\bod{\bold d}
\nc\boe{\bold e}
\nc\bof{\bold f}
\nc\bog{\bold g}
\nc\boh{\bold h}
\nc\boi{\bold i}
\nc\boj{\bold j}
\nc\bok{\bold k}
\nc\bol{\bold l}
\nc\bom{\bold m}
\nc\bon{\bold n}
\nc\boo{\bold o}
\nc\bop{\bold p}
\nc\boq{\bold q}
\nc\bor{\bold r}
\nc\bos{\bold s}
\nc\bou{\bold u}
\nc\bov{\bold v}
\nc\bow{\bold w}
\nc\boz{\bold z}

\nc\ba{\bold A}
\nc\bb{\bold B}
\nc\bc{\bold C}
\nc\bd{\bold D}
\nc\be{\bold E}
\nc\bg{\bold G}
\nc\bh{\bold H}
\nc\bi{\bold I}
\nc\bj{\bold J}
\nc\bk{\bold K}
\nc\bl{\bold L}
\nc\bm{\bold M}
\nc\bn{\bold N}
\nc\bo{\bold O}
\nc\bp{\bold P}
\nc\bq{\bold Q}
\nc\br{\bold R}
\nc\bs{\bold S}
\nc\bt{\bold T}
\nc\bu{\bold U}
\nc\bv{\bold V}
\nc\bw{\bold W}
\nc\bz{\bold Z}
\nc\bx{\bold X}

\nc\ca{\mathcal A}
\nc\cb{\mathcal B}
\nc\cc{\mathcal C}
\nc\cd{\mathcal D}
\nc\ce{\mathcal E}
\nc\cf{\mathcal F}
\nc\cg{\mathcal G}
\rnc\ch{\mathcal H}
\nc\ci{\mathcal I}
\nc\cj{\mathcal J}
\nc\ck{\mathcal K}
\nc\cl{\mathcal L}
\nc\cm{\mathcal M}
\nc\cn{\mathcal N}
\nc\co{\mathcal O}
\nc\cp{\mathcal P}
\nc\cq{\mathcal Q}
\nc\car{\mathcal R}
\nc\cs{\mathcal S}
\nc\ct{\mathcal T}
\nc\cu{\mathcal U}
\nc\cv{\mathcal V}
\nc\cz{\mathcal Z}
\nc\cx{\mathcal X}
\nc\cy{\mathcal Y}

\nc\e[1]{E_{#1}}
\nc\ei[1]{E_{\delta - \alpha_{#1}}}
\nc\esi[1]{E_{s \delta - \alpha_{#1}}}
\nc\eri[1]{E_{r \delta - \alpha_{#1}}}
\nc\ed[2][]{E_{#1 \delta,#2}}
\nc\ekd[1]{E_{k \delta,#1}}
\nc\emd[1]{E_{m \delta,#1}}
\nc\erd[1]{E_{r \delta,#1}}

\nc\ef[1]{F_{#1}}
\nc\efi[1]{F_{\delta - \alpha_{#1}}}
\nc\efsi[1]{F_{s \delta - \alpha_{#1}}}
\nc\efri[1]{F_{r \delta - \alpha_{#1}}}
\nc\efd[2][]{F_{#1 \delta,#2}}
\nc\efkd[1]{F_{k \delta,#1}}
\nc\efmd[1]{F_{m \delta,#1}}
\nc\efrd[1]{F_{r \delta,#1}}

\nc\fa{\frak a}
\nc\fb{\frak b}
\nc\fc{\frak c}
\nc\fd{\frak d}
\nc\fe{\frak e}
\nc\ff{\frak f}
\nc\fg{\frak g}
\nc\fh{\frak h}
\nc\fj{\frak j}
\nc\fk{\frak k}
\nc\fl{\frak l}
\nc\fm{\frak m}
\nc\fn{\frak n}
\nc\fo{\frak o}
\nc\fp{\frak p}
\nc\fq{\frak q}
\nc\fr{\frak r}
\nc\fs{\frak s}
\nc\ft{\frak t}
\nc\fu{\frak u}
\nc\fv{\frak v}
\nc\fz{\frak z}
\nc\fx{\frak x}
\nc\fy{\frak y}

\nc\fA{\frak A}
\nc\fB{\frak B}
\nc\fC{\frak C}
\nc\fD{\frak D}
\nc\fE{\frak E}
\nc\fF{\frak F}
\nc\fG{\frak G}
\nc\fH{\frak H}
\nc\fJ{\frak J}
\nc\fK{\frak K}
\nc\fL{\frak L}
\nc\fM{\frak M}
\nc\fN{\frak N}
\nc\fO{\frak O}
\nc\fP{\frak P}
\nc\fQ{\frak Q}
\nc\fR{\frak R}
\nc\fS{\frak S}
\nc\fT{\frak T}
\nc\fU{\frak U}
\nc\fV{\frak V}
\nc\fZ{\frak Z}
\nc\fX{\frak X}
\nc\fY{\frak Y}
\nc\tfi{\ti{\Phi}}
\nc\bF{\bold F}
\rnc\bol{\bold 1}

\nc\ua{\bold U_\A}

\nc\qinti[1]{[#1]_i}
\nc\q[1]{[#1]_q}
\nc\xpm[2]{E_{#2 \delta \pm \alpha_#1}}  
\nc\xmp[2]{E_{#2 \delta \mp \alpha_#1}}
\nc\xp[2]{E_{#2 \delta + \alpha_{#1}}}
\nc\xm[2]{E_{#2 \delta - \alpha_{#1}}}
\nc\hik{\ed{k}{i}}
\nc\hjl{\ed{l}{j}}
\nc\qcoeff[3]{\left[ \begin{smallmatrix} {#1}& \\ {#2}& \end{smallmatrix}
\negthickspace \right]_{#3}}
\nc\qi{q}
\nc\qj{q}

\nc\ufdm{{_\ca\bu}_{\rm fd}^{\le 0}}

\nc\isom{\cong} 

\nc{\pone}{{\Bbb C}{\Bbb P}^1}
\nc{\pa}{\partial}

\nc{\F}{{\mathcal F}}
\nc{\Sym}{{\goth S}}
\nc{\A}{{\mathcal A}}
\nc{\arr}{\rightarrow}
\nc{\larr}{\longrightarrow}

\nc{\ri}{\rangle}
\nc{\lef}{\langle}
\nc{\W}{{\mathcal W}}
\nc{\uqatwoatone}{{U_{q,1}}(\su)}
\nc{\uqtwo}{U_q(\goth{sl}_2)}
\nc{\dij}{\delta_{ij}}
\nc{\divei}{E_{\alpha_i}^{(n)}}
\nc{\divfi}{F_{\alpha_i}^{(n)}}
\nc{\Lzero}{\Lambda_0}
\nc{\Lone}{\Lambda_1}
\nc{\ve}{\varepsilon}
\nc{\phioneminusi}{\Phi^{(1-i,i)}}
\nc{\phioneminusistar}{\Phi^{* (1-i,i)}}
\nc{\phii}{\Phi^{(i,1-i)}}
\nc{\Li}{\Lambda_i}
\nc{\Loneminusi}{\Lambda_{1-i}}
\nc{\vtimesz}{v_\ve \otimes z^m}

\nc{\asltwo}{\widehat{\goth{sl}_2}}
\nc\ag{\widehat{\goth{g}}}  
\nc\teb{\tilde E_\boc}
\nc\tebp{\tilde E_{\boc'}}

\newcommand{\eeq}{\end{equation}}
\newcommand{\ben}{\begin{eqnarray}}
\newcommand{\een}{\end{eqnarray}}

\begin{document}

\title[$q-$Onsager algebra and $q-$special functions]
{The $q-$Onsager algebra and  multivariable \\ $q-$special functions}

\author{Pascal Baseilhac$^{\dagger,*}$}
\address{$^\dagger$Laboratoire de Math\'ematiques et Physique Th\'eorique CNRS/UMR 7350,
 F\'ed\'eration Denis Poisson FR2964,
Universit\'e de Tours,
Parc de Grammont, 37200 Tours, 
FRANCE}
\email{baseilha@lmpt.univ-tours.fr}

\author{Luc Vinet$^{*}$}
\address{$^{*}$ Centre de recherches math\'ematiques Universit\'e de Montr\'eal, CNRS/UMI 3457, P.O. Box 6128, Centre-ville Station, Montr\'eal (Qu\'ebec), H3C 3J7 CANADA}
\email{vinet@CRM.UMontreal.CA}

\author{Alexei Zhedanov$^{\diamond,*,\times}$}
\address{$^{\diamond}$ Department of Mathematics, Information School, Renmin University of China, Beijing 100872,
CHINA}
\address{$^{\times}$ Donetsk Institute for Physics and Technology, Donetsk 83114, UKRAINE}
\email{zhedanov@yahoo.com}

\date{August, 2017}

\begin{abstract}
Two sets of mutually commuting  $q-$difference operators $x_i$ and $y_j$, $i,j=1, ...,N$ such that
$x_i$ and $y_i$ generate a homomorphic image of the $q-$Onsager algebra for each $i$ are introduced.
The common polynomial eigenfunctions of each set are found to be entangled product of elementary Pochhammer functions in $N$ variables and $N+3$ parameters. Under certain conditions on the parameters, they form  two `dual' bases of polynomials in $N$ variables. The action of each operator with respect to its dual basis is block tridiagonal. The overlap coefficients between the two dual bases are expressed as entangled products of $q-$Racah polynomials and satisfy an orthogonality relation.
The overlap coefficients between either one of these bases and the multivariable monomial basis are also considered. One obtains in this case entangled products of dual $q-$Krawtchouk polynomials. Finally, the `split' basis in which the two families of operators act as block bidiagonal matrices is also provided.
\end{abstract}

\maketitle

\vskip -0.5cm

{\small MSC:\ 81R50;\ 81R10;\ 81U15;\ 39A70;\ 33D50;\ 39A13.}

{{\small  {\it \bf Keywords}:  $q-$Onsager algebra;  Multivariable polynomials; Orthogonality; Tridiagonal pairs}}

\vspace{5mm}



\section{Introduction}
In view of the intrinsic mathematical interest from the representation theoretic viewpoint and of the various applications in physical models, much attention has been devoted to the algebraic underpinning of multivariate special functions and orthogonal polynomials. The connection between double affine Hecke algebras  or Cherednik algebras and Macdonald--Koornwinder polynomials has proved to be very fruitful \cite{1}. The search for a similar interpretation of multivariate polynomials of the Tratnik type and their $q-$analogs \cite{2,3,4,Iliev,BM} has been initiated lately. \vspace{1mm}

It has been appreciated in the univariate case that the Askey--Wilson or Zhedanov algebra \cite{7} with the Bannai--Ito algebra \cite{8,9} and Racah algebra \cite{10} as special cases\footnote{The defining relations of the Askey--Wilson algebra are in terms of a scalar parameter $q$. The Bannai--Ito and Racah algebras correspond to the specialization $q^2=-1$  and $q^2=1$, respectively.}, encodes the bispectrality properties of the corresponding polynomials.  Different generalizations of these three algebras have been recently considered in order to tackle multivariate extensions. Regarding the two special cases, extensions of the Bannai--Ito and Racah algebras have been introduced and studied \cite{12,13}. Constructions rely on the tensorial products  of the two Lie (super)algebras i.e. $osp(1|2)$ and $sl_2$, respectively. Bases for the generalized Bannai--Ito and Racah algebras' modules have been explicitly constructed and the overlap coefficients between these bases have been seen to be expressed in terms of the corresponding multivariate polynomials. While a construction of a higher rank Askey--Wilson algebra along those lines is still awaited, an alternative framework to extend the algebraic picture to the multivariate realm for $q$ not a root of unity  is provided by the $q-$Onsager algebra \cite{Ter03,Bas2}. In this last context, infinite and finite dimensional modules of the $q-$Onsager algebra have been constructed in terms of the multivariate Gasper--Rahman polynomials, in \cite{BM} $N-$pairs of operators generating the $q-$Onsager algebra were related to Iliev's families of $q-$difference operators \cite{Iliev}.  Another approach to connect $q-$special functions to the $q-$Onsager algeba is considered here. It hinges on the fact that the $q-$Onsager algebra is a coideal subalgebra of $U_q(\widehat{sl_2})$ \cite{Bas2,Bas3} (see also \cite{Ko}) and that tensor product representations of this algebra can thus be expected to lead  to multivariate extensions of the Askey--Wilson polynomials.\vspace{1mm}

There is a rather large class of quantum integrable models in the continuum or on lattices 
whose local integrals of motion can be written in terms of the elements
of an Abelian subalgebra of the $q-$Onsager algebra \cite{Bas2}. This is so for instance in the case of the two-dimensional Ising model \cite{Ons}, of the 
superintegrable Potts model \cite{Potts} at $q=1$ or of the open XXZ spin chain for $q\neq 1$ \cite{BK2,BB3}.
For the integrable models that fall in this class, finding the spectrum and eigenstates of the Hamiltonian
relies on the construction of explicit finite or infinite dimensional representations of the $q-$Onsager algebra. We here show that in this
algebraic framework, a formulation of the Hamiltonian's eigenfunctions in terms of multivariable $q-$special functions is possible - an observation that is bound to prove quite fruitful in the analysis of those models.\vspace{1mm}

In the following paper, we construct explicitly three different types of bases for the $q-$Onsager algebra in terms of multivariate $q-$special functions. Two bases are  such that the $q-$Onsager algebra's generators act as diagonal or block tridiagonal matrices. In the third basis, the so-called `split' basis, they act as upper or lower bidiagonal matrices.  \vspace{1mm}  

The paper is structured as follows. In Section 2, the definition of the $q-$Onsager algebra  (see Definition \ref{Oq}) and the action of the two fundamental generators on tensor products of $U_{q}(\widehat{sl_2})$ evaluation representations  are recalled. In Section 3, using the $q-$difference operator realization of $U_{q}(sl_2)$ \cite{S1},  two families of mutually commuting $q-$difference operators in $i-$variables $\{z_1,z_2,...,z_i\}$, $i=1,...,N$, that generate the $q-$Onsager algebra,  are given in Proposition \ref{realqdiff}. Their respective eigenfunctions are found to be entangled products of elementary Pochhammer functions in the $N$ variables with $N+3$ additional parameters. These functions form two `dual' bases of the polynomial vector space, see Propositions \ref{eigenF}, \ref{eigenF2} and Lemma \ref{lem-basis}. Namely, the common eigenfunctions of $\cW_0^{(i)}$, $i=1,2,...,N$, can be written as:
\beqa
\qquad F^{(N)}_{\{n\}} (z_1,z_2,...,z_N)=\prod_{i=1}^{N}z_i^{2j_i} 
\left(z^{(i)}_-/z_i;q^{-2}\right)_{n_i}\left(z^{(i)}_+/z_i;q^{-2}\right)_{2j_i-n_i}.\label{Fpoc}
\eeqa
with (\ref{poch}), $j_i\in\frac{1}{2}{\mathbb N}$ and $n_i\in \{0,1,...,2j_i\}$. The `dual' eigenfunctions $\tilde{F}^{(N)}_{\{\tilde{n}\}}(\{z\})$  associated with $\cW_1^{(i)}$ are obtained through the substitutions:
\beqa
n_i \rightarrow \tilde{n}_i,\quad z^{(i)}_\pm \rightarrow \tilde{z}^{(i)}_\pm ,\quad q \rightarrow q^{-1},\nonumber 
\eeqa
with (\ref{zpmi}), (\ref{ztildepmi}). Then, as stated in Proposition \ref{proptrid} it is shown that the action of  $\cW_1^{(i)}$ (resp. $\cW_0^{(i)}$) in the eigenbasis of $\cW_0^{(i)}$ (resp. $\cW_1^{(i)}$) is `block' tridiagonal.  The cases $N=1$ and $N=2$ are described in details. In Section 4,  the overlap coefficients between the two dual bases 
or between any of the two bases and  the multivariable monomial basis are identified. They are written as entangled products of the $q-$Racah polynomials and dual $q-$Krawtchouk polynomials, respectively. In Section 5, another basis generalizing the (one-variable) `split' basis \cite[Remark 2.1]{Ro2} is provided. It has for elements:
\beqa
\qquad G^{(N)}_{\{n\}} (z_1,z_2,...,z_N)=\prod_{i=1}^{N}z_i^{2j_i} 
\left(z^{(i)}_-/z_i;q^{-2}\right)_{n_i}\left(\overline{z}^{(i)}_+/z_i;q^{2}\right)_{2j_i-n_i}\label{Gpoc}
\eeqa
with (\ref{zbar}) defining the parameters. It is shown that $\cW_0^{(i)}$ (resp. $\cW_1^{(i)}$) act as upper (resp. lower) block bidiagonal matrices in this basis. Concluding remarks will be found in the last section. Three appendices supplement the text. \vspace{1mm}

Let us mention that the subject of this paper is closely related with the theory of tridiagonal pairs \cite{Ter01}. Indeed, from the point of view of representation theory, all orthogonal polynomials in one variable can be cast in the framework
of Leonard pairs \cite{Ter2}. Similarly, our results find a natural interpretation in terms of tridiagonal pairs \cite{Ter01}: provided the polynomial vector space is irreducible, the $q-$difference operators $\cW_0^{(i)}$, $\cW_1^{(i)}$, $i=1,2,...,N$ offer examples of tridiagonal pairs of $q-$Racah type (see \cite{Ter01,Ter03} and related works). In this context, the three bases constructed here generate explicitly the two examples of the two sets of eigenspaces defined in \cite[Definition 2.1]{Ter01} and of the `split' subspaces introduced in \cite[Theorem 4.6]{Ter01}.\vspace{1mm}

\subsection{Notations}
In this paper, we fix a nonzero complex number $q$ which is not a root of unity.  We will use the standard $q$-shifted factorials (also called $q-$Pochhammer functions) \cite{KS}:
\beqa
(a;q)_n=\prod_{k=0}^{n-1}(1-aq^{k}), \quad
(a_1,a_2,\dots,a_k;q)_n=\prod_{j=1}^k(a_j;q)_n.\label{poch}
\eeqa

Let $\mathcal{P}^{(N)}_z={\mathbb C}[ z_1,z_2,...,z_N]$ be the vector space of  polynomials of total degree $2\kJ_N$ in the variables $z_1,z_2,...,z_N$. We denote the $q$-shift difference operators in the $j$-th variable acting on a function $f(z)\equiv f(z_1,z_2,\dots, z_N)$ as:
\begin{eqnarray*}
T_{\pm}^{(i)} f(z)=f(z_1,z_2,\dots,q^{\pm 1}z_i,\dots,z_N).\nonumber
\end{eqnarray*}
\vspace{1mm}

\section{The $q-$Onsager algebra and $q-$difference operators}
In this section, we first introduce the $q-$Onsager algebra through generators and relations \cite{Ter03,Bas2}. We recall how it is embedded into $U_{q}(\widehat{sl_2})$ as a certain coideal subalgebra \cite{Bas3} and describe the action of the generators on tensor product of evaluation representations. Using the $q-$difference realization of $U_{q}(sl_2)$ (see Appendix \ref{A3}), we shall obtain two families of mutually commuting  $q-$difference operators in $i-$variables.
\vspace{1mm}
\begin{defn}[\cite{Ter03,Bas2}] \label{Oq}
Let $\rho$ be a complex scalar.  The $q-$Onsager algebra $O_q(\widehat{sl_2})$ is the associative algebra with unit and standard generators $\textsf{W}_0,\textsf{W}_1$ subject to the relations\footnote{The $q-$commutator $\big[X,Y\big]_q=qXY-q^{-1}YX$, where $q$ is called the deformation parameter, is introduced.}
\beqa
[\textsf{W}_0,[\textsf{W}_0,[\textsf{W}_0,\textsf{W}_1]_q]_{q^{-1}}]=\rho[\textsf{W}_0,\textsf{W}_1]\
,\qquad
[\textsf{W}_1,[\textsf{W}_1,[\textsf{W}_1,\textsf{W}_0]_q]_{q^{-1}}]=\rho[\textsf{W}_1,\textsf{W}_0]\
\label{qDG} . \eeqa
\end{defn}
\begin{rem} For $\rho=0$ the relations (\ref{qDG})
reduce to the $q-$Serre relations of $U_{q}(\widehat{sl_2})$. For $q=1$, $\rho=16$ they coincide with the Dolan-Grady relations \cite{DG}.
\end{rem}

The $q-$Onsager algebra is known to be isomorphic\footnote{For the proof of isomorphism, see \cite{Ko}.} to a certain coideal subalgebra of $U_{q}(\widehat{sl_2})$ \cite{Bas3}. Introduce the Chevalley generators  $E_i,F_i,q^{H_i}$, $i=0,1$ of $U_{q}(\widehat{sl_2})$, see Appendix A. 
\begin{prop}[\cite{Bas3}]
\label{iso}
Let $\{k_\pm,\epsilon_\pm\}$ be complex scalars. There is an algebra morphism $O_{q}(\widehat{sl_2}) \mapsto U_{q}(\widehat{sl_2})$ such that
\beqa
{\textsf W}_0 &\mapsto& k_+E_1q^{H_1/2} + k_-F_1q^{H_1/2}+\epsilon_+ q^{H_1},\label{defW}\\
{\textsf W}_1 &\mapsto& k_-E_0q^{H_0/2} + k_+F_0q^{H_0/2}+\epsilon_- q^{H_0}\nonumber
\eeqa
with
\beqa
\rho=(q+q^{-1})^2k_+k_-\ .\label{rho}
\eeqa
\end{prop}

The action of the generators of the $q-$Onsager algebra on tensor product representations can be considered as follows. To this end, the concept of coaction map \cite{CP} is needed.
\begin{prop}[\cite{BSh1}]\label{coac} 
Let $k_\pm$ be complex scalars and take $\rho$ as in (\ref{rho}). The $q-$Onsager algebra $O_{q}(\widehat{sl_2})$ is a left $U_{q}(\widehat{sl_2})-$comodule algebra with coaction map $\delta: O_{q}(\widehat{sl_2})\rightarrow U_{q}(\widehat{sl_2})\otimes O_{q}(\widehat{sl_2})$
such that
\beqa
\delta({\textsf W}_0)&=& (k_+E_1q^{H_1/2} + k_-F_1q^{H_1/2})\otimes 1 + q^{H_1} \otimes {\textsf W}_0,\label{deltadef}\\
\delta({\textsf W}_1)&=& (k_-E_0q^{H_0/2} + k_+F_0q^{H_0/2})\otimes 1 + q^{H_0} \otimes {\textsf W}_1.\ \nonumber
\eeqa
\end{prop}
\begin{rem} Considering the embedding of the $q-$Onsager algebra into  $U_{q}(\widehat{sl_2})$ of Proposition \ref{iso}, the coaction map is identified with the coproduct of  $U_{q}(\widehat{sl_2})$ (see Appendix \ref{A1}). 
\end{rem}

Finite dimensional (evaluation) representations of the $q-$Onsager algebra are now used to construct two families of mutually commuting operators. Let us denote $\{S_\pm,q^{s_3}\}$ as the generators of the algebra $U_q(sl_2)$ with defining relations (\ref{Uqsl2}). From Appendix \ref{A2} and Proposition \ref{coac}, it follows: 
\begin{prop}[See \cite{Bas3}]
\label{hom}
Let $\{k_+,k_-,\epsilon_\pm\}$ be complex scalars. Let $\{v_i|i=1,2,...,N\}$ denote the nonzero evaluation parameters. Define:
\beqa
{\cal W}_{0}^{(i)}&=&
\left(k_+v_iq^{1/2}S_+q^{s_3}
+k_-v_i^{-1}q^{-1/2}S_-q^{s_3}\right)\otimes I\!\!I^{(i-1)}\ +\ q^{2s_3}\otimes {\cal W}_{0}^{(i-1)}
\ , \label{defWloopN}\\
{\cal W}_{1}^{(i)}&=&
\left(k_+v^{-1}_iq^{-1/2}S_+q^{-s_3} + k_-v_iq^{1/2}S_-q^{-s_3}\right)\otimes I\!\!I^{(i-1)}\ +\ q^{-2s_3}\otimes {\cal W}_{1}^{(i-1)}
\ \nonumber
\eeqa
with \ ${{\cal W}}_{0}^{(0)}\equiv \epsilon_+\ , {{\cal W}}_{1}^{(0)}\equiv \epsilon_-$. For any $i=0,1,2...,N$, one has the homomorphism:
\beqa
{\textsf W}_0  \mapsto {\cal W}_{0}^{(i)}\ ,\qquad{\textsf W}_1  \mapsto  {\cal W}_{1}^{(i)}\nonumber
\eeqa
with (\ref{rho}).
\end{prop}
\begin{lem} The operators $\cW_0^{(i)}$ (resp. $\cW_1^{(i)}$) are mutually commuting:
\beqa
 \big[\cW_0^{(i)}, \cW_0^{(j)}\big]=0 \qquad \mbox{and}\qquad   \big[\cW_1^{(i)}, \cW_1^{(j)}\big]=0 \quad \mbox{for any}\quad i,j=1,2,...,N.
\eeqa
\end{lem}
\begin{proof}
By induction, use (\ref{defWloopN}).
\end{proof}

Finite dimensional tensor product representations of the $q-$Onsager algebra can be realized by $q-$difference operators acting in the linear space of multivariable polynomials of total degree $2(j_1+j_2+...+j_N)$, $j_i\in \frac{1}{2}{\mathbb N}$, in the variables $z_1,z_2,...,z_N$. From Appendix \ref{A3} and Proposition \ref{coac}, the following proposition is obtained:
\begin{prop}\label{realqdiff} On $\mathcal{P}^{(N)}_z$, the operators $\cW_0^{(i)}$ and $\cW_1^{(i)}$ act as $q-$difference operators of the form:
\beqa
 {\cal W}_{0}^{(i)}&\mapsto& \sum_{k=1}^{i} q^{-2(j_i+j_{i-1}+...+j_{k+1})} {T_+^{(i)}}^2{T_+^{(i-1)}}^2\cdots {T_+^{(k+1)}}^2 \left( b_0^{(k)} z_k (q^{2j_k} -q^{-2j_k} {T_+^{(k)}}^2)+ c_0^{(k)}  z_k^{-1}(1- {T_+^{(k)}}^2)\right)\nonumber\\
&&\qquad \qquad  \qquad\ + \ \epsilon_+ q^{-2(j_i+j_{i-1}+...+j_{1})} {T_+^{(i)}}^2{T_+^{(i-1)}}^2\cdots {T_+^{(1)}}^2 ,\label{qdiffW0gen}
\eeqa
and similarly for ${\cal W}_{1}^{(i)}$ with the substitution $q\rightarrow q^{-1},b^{(k)}_0 \rightarrow  - b^{(k)}_1 , c^{(k)}_0\rightarrow -c^{(k)}_1$, $\epsilon_+\rightarrow \epsilon_-$ and  \ $T_+^{(k)}\rightarrow T_-^{(k)}$, where:
\beqa
b_0^{(i)} &=& \frac{k_+v_iq^{1/2-j_i}}{(q-q^{-1})}, \qquad c_0^{(i)} = - \frac{k_-v_i^{-1}q^{-1/2-j_i}}{(q-q^{-1})},\label{coeffbc}\\
b_1^{(i)} &=& \frac{k_+v_i^{-1}q^{-1/2+j_i}}{(q-q^{-1})}, \qquad c_1^{(i)} = - \frac{k_-v_iq^{1/2+j_i}}{(q-q^{-1})}.\nonumber
\eeqa
\end{prop}
\vspace{1mm}

Note that the multivariable $q-$difference operator realization of the $q-$Onsager algebra of Proposition \ref{realqdiff} is not the same as the one recently proposed in \cite[Proposition 2.3]{BM}.

\vspace{3mm}

\section{Two `dual' polynomial eigenbases}

In this section, two `dual' multivariable polynomial bases for $\mathcal{P}^{(N)}_z$ are explicitly constructed. In a first part, it is shown that the first (resp. second) multivariable $q-$difference operators given in Proposition \ref{realqdiff}  are simultaneously diagonalized by any of the basic multivariable polynomial eigenvectors from the first (resp. second) basis, see Propositions \ref{eigenF}, \ref{eigenF2} and Lemma \ref{lem-basis}.  In a second part, we study the action of the second set of $q-$difference operators on the eigenbasis of the first set, see Proposition \ref{proptrid}. The cases $N=1$ and $N=2$ are described in details.  \vspace{1mm}

Basically, for any positive integer $N$   we first solve the following spectral problem\footnote{The analysis presented here can be understood as a multivariable generalization of the analysis in \cite[Subsection 3.3.2]{WZ} and \cite[Remark 2.1]{Ro2}. For the special case $N=1$, one recovers the results of \cite{WZ,Ro2}.} on  $\mathcal{P}^{(N)}_z$:
\beqa
{\cal W}_{0}^{(i)} F^{(N)}_{\{n\}} (z_1,z_2,...,z_N)  &=&  \lambda_{\{n\}}^{(i)}F^{(N)}_{\{n\}} (z_1,z_2,...,z_N)\label{specN},\\
{\cal W}_{1}^{(i)} \tilde{F}^{(N)}_{\{\tilde{n}\}} (z_1,z_2,...,z_N)  &=&  \tilde{\lambda}_{\{\tilde{n}\}}^{(i)} \tilde{F}^{(N)}_{\{\tilde{n}\}} (z_1,z_2,...,z_N)\qquad \mbox{for}\qquad i=1,2,...,N,\label{specN2}
\eeqa
where $\{n\}=\{n_1,n_2,...,n_N\}$, $n_i\in \{0,1,...,2j_i\}$, and similarly for $\{\tilde{n}\}$.
\vspace{1mm} 

Without loss of generality and for further convenience,  let us introduce the parametrization:
\beqa
k_+=-\frac{(q-q^{-1})}{2}q^\eta,\quad  k_-=\frac{(q-q^{-1})}{2}q^{\eta'},\quad \epsilon_+ = q^{\frac{\eta+\eta'}{2}}\cosh\alpha,\quad  \epsilon_- = q^{\frac{\eta+\eta'}{2}}\cosh\alpha^* ,\label{param}
\eeqa
where $\eta,\eta',\alpha,\alpha^*$ are arbitrary complex scalars. We shall also use the notation $\kJ_i = j_1+j_2+...+j_i$ and $\kN_i = n_1+n_2+...+n_i$, $\kJ_0 =\kN_0 =0$.

\begin{prop}\label{eigenF} The solution of the spectral problem (\ref{specN}) is given by:
\beqa
\qquad F^{(N)}_{\{n\}}(z_1,z_2,...,z_N) &=& \prod_{i=1}^{N} f_{n_i}^{(i)}(z_i) \quad \mbox{with} \quad 
 f_{n_i}^{(i)}(z_i)=\prod_{k=0}^{2j_i-1-n_i}   (z_i - z_{+}^{(i)}q^{-2k})   \prod_{l=0}^{n_i-1} (z_i - z_{-}^{(i)}q^{-2l})\label{FN}
\eeqa
and
\beqa
\lambda_{n_1,n_2,...,n_i}^{(i)} &=& \frac{1}{2} q^{\frac{\eta+\eta'}{2}}\left(   e^\alpha q^{-2\kJ_i+2\kN_i} +  e^{-\alpha} q^{2\kJ_i-2\kN_i}\right)\quad \mbox{for} \quad i=1,2,...,N,\label{spec1}
\eeqa
where
\beqa
z_\pm^{(i)}  = - v_i^{-1}q^{-1/2 +j_i +\frac{\eta'-\eta}{2}}e^{\pm\alpha}q^{\mp 2(\kJ_{i-1} -\kN_{i-1})}\quad \mbox{for} \quad i=1,2,...,N.\label{zpmi}
\eeqa
\end{prop}
\begin{proof}
Consider (\ref{specN}) for the case $N=1$. The corresponding spectral problem is of the form $\cW_0^{(1)} f(z_1)= \lambda f(z_1)$ where $\lambda$ is a scalar. According to Proposition \ref{realqdiff}, it leads to a first-order $q-$difference equation with respect to the shift $z_1\rightarrow q^{2}z_1$:
\beqa
a^{(1)}(z_1) f(q^2z_1)    + (u^{(1)}(z_1) -\lambda) f(z_1) =0\label{TQ1}
\eeqa
where the Laurent polynomials $a^{(1)}(z_1),u^{(1)}(z_1) $ are respectively given by:
\beqa
a^{(1)}(z_1)  = -b_0^{(1)}q^{-2j_1} z_1 - c_0^{(1)}z_1^{-1} + \epsilon_+q^{-2j_1}, \quad u^{(1)}(z_1) = b_0^{(1)} q^{2j_1}z_1 +   c_0^{(1)}z_1^{-1}.\nonumber
\eeqa
Assume $f(z_1)$ is a polynomial of maximal degree $2j_1$, factorized as:
\beqa
f(z_1)=\prod_{k=1}^{2j_1}(z_1-\xi_k).\label{solf}
\eeqa
Inserting (\ref{solf}) into (\ref{TQ1}), one finds that the roots $\{\xi_k|k=1,2,...,2j_1\}$ of the polynomial must satisfy\footnote{In the `physics' literature, this set of equations is often called `Bethe equations'.}:
\beqa
a^{(1)}(\xi_l)\prod_{k=1,k\neq l}^{2j_1}(q^2\xi_l-\xi_k)=0, \qquad l=1,2,...,2j_1.
\eeqa
Let $z_\pm$ denote the two roots of the Laurent polynomial $a^{(1)}(z_1)$. Then, there are exactly $2j_1+1$  polynomial eigenfunctions of the form (\ref{solf}) that can be constructed and they are given by (\ref{FN}) for $N=1$. Finally, substituting $f_{n_1}^{(1)}(z_1)$ into (\ref{TQ1}) and equating the constant terms, one finds:
\beqa
\lambda_{n_1}^{(1)} &=& b_0^{(1)}\left(z_{+}^{(1)} q^{-2j_1+2n_1} + z_{-}^{(1)} q^{2j_1-2n_1}\right).\nonumber
\eeqa 
Using the parametrization (\ref{param}), one arrives at (\ref{spec1}) for $N=1$. Note that the arguments presented here can be found in e.g.\cite[Subsection 3.3.2]{WZ}. \vspace{1mm}

Consider the case $N=2$ in (\ref{specN}). Take an eigenfunction of the factorized form $F^{(2)}_{n_1,n_2}(z_1,z_2)=f(z_2)F^{(1)}_{n_1}(z_1)$.
It follows from  (\ref{qdiffW0gen}) that $f(z_2)$ has to solve the auxiliary spectral problem:
\beqa
\cW_0^{(1)}|_{\epsilon_+\rightarrow \lambda_{n_1}^{(1)},z_1\rightarrow z_2, j_1\rightarrow j_2,  b_0^{(1)}\rightarrow  b_0^{(2)},  c_0^{(1)}\rightarrow  c_0^{(2)}} f(z_2) = \lambda_{n_1,n_2}^{(2)}f(z_2).\nonumber
\eeqa
Applying the same analysis as for the case $N=1$, the claim follows for $N=2$. The proof of Proposition \ref{eigenF} for generic values of $N$ is then completed by an inductive argument.
\end{proof}

Either by using symmetry relations between $\cW_0^{(i)}$ and $\cW_1^{(i)}$ (see Proposition (\ref{realqdiff}) or by applying to the spectral problem (\ref{specN2}) an analysis similar to the one presented above, the second set of eigenfunctions is analogously derived:
\begin{prop}\label{eigenF2} The solution of the spectral problem (\ref{specN2}) is given by:
\beqa
\qquad \tilde{F}^{(N)}_{\{\tilde{n}\}}(z_1,z_2,...,z_N) = F^{(N)}_{\{n\}}(z_1,z_2,...,z_N)|_{n\rightarrow \tilde{n},q\rightarrow q^{-1},z_\pm^{(i)}\rightarrow \tilde{z}_\pm^{(i)}, i=1,2,...,N}\nonumber
\eeqa
and
\beqa
\tilde{\lambda}_{\tilde{n}_1,\tilde{n}_2,...,\tilde{n}_i}^{(i)} &=& \frac{1}{2} q^{\frac{\eta+\eta'}{2}}\left(  e^{-\alpha^*} q^{-2\kJ_i+2\tilde{\kN}_i} +  e^{\alpha^*} q^{2\kJ_i-2\tilde{\kN}_i}\right)\label{speclt}
\eeqa
where
\beqa
\tilde{z}_\pm^{(i)}  =  v_iq^{1/2 -j_i +\frac{\eta'-\eta}{2}}e^{\pm\alpha^*}q^{\pm 2(\kJ_{i-1}-\tilde{\kN}_{i-1})}.\label{ztildepmi}
\eeqa
\end{prop}
\begin{rem} Note that the scalar parameter $\eta'$ can be removed. In total, there remains $N+3$ free scalar parameters: $\alpha,\alpha^*,\eta,v_1,v_2,...,v_N$.
\end{rem}

Observe that above eigenfunctions can be written in terms of $q-$Pochhammer symbols to obtain (\ref{Fpoc}). The following Lemma (and its proof) can be viewed as multivariate extensions of \cite[Lemma 3.1]{Ro2}.
\begin{lem}\label{lem-basis} Assume
\beqa
&&z_+^{(i)}/z_-^{(i)}\notin\{q^{4j_i-2},...,q^{-4j_i+2}\}, \quad z_\pm^{(i)}\neq 0,\quad z_\pm^{(i)}\neq z_\pm^{(j)},\label{cond1zz}\\
&&\tilde{z}_+^{(i)}/\tilde{z}_-^{(i)}\notin\{q^{4j_i-2},...,q^{-4j_i+2}\},\quad \tilde{z}_\pm^{(i)}\neq 0,\quad \tilde{z}_\pm^{(i)}\neq \tilde{z}_\pm^{(j)}\label{cond2zz}
\eeqa
for any $i\neq j$. The vector space $\mathcal{P}^{(N)}_z$ admits two `dual' bases.  The first basis is generated by the polynomials $\{F^{(N)}_{\{n\}}(z_1,z_2,...,z_N)\}$ and the second basis by $\{\tilde{F}^{(N)}_{\{\tilde{n}\}}(z_1,z_2,...,z_N)\}$. The cardinality is given by $\prod_{k=1}^N(2j_k+1)$.
\end{lem}
\begin{proof} We study the conditions under which the polynomials $\mathcal{B}_z=\{
F_{\{n\}}^{(N)}(z_1, z_2,\cdots,z_N)\}$ form a basis of
$\mathcal{P}^{(N)}_z$. Let us start with $N=1$. Since the set of functions $\mathcal{B}_z$ has
a cardinality which coincides with the dimension of the linear space, we just have to
show:
\begin{equation}
\sum_{n_1=0}^{2j_1} \zeta_{n_1} F^{(1)}_{n_1}(z_1)=0\quad  \mbox{iff} \quad
\zeta_{n_1}=0 \quad \mbox{for all} \quad n_1.\label{lin} \end{equation}
Note that $F^{(1)}_{n_1}(z_1)$ is a polynomial of degree $2j_1$ in the variable $z_1$, so it is sufficient to
check the equation on $2j_1+1$ distinct values of $z_1$.
First, suppose that $z_\pm^{(1)}\neq 0$. Observe that all polynomials $F^{(1)}_{n_1}(z_1)$ have the common zero $z_1=z_-^{(1)}$ if $z_+^{(1)}/z_-^{(1)}=q^{2k}$ or $z_1=z_+^{(1)}$ if $z_+^{(1)}/z_-^{(1)}=q^{-2k}$ with $k=0,1,...,2j_1-1$. So, we assume that all conditions in (\ref{cond1zz}) are satisfied. Choosing $z_1=z_-^{(1)}$ in (\ref{lin}), one finds $\zeta_{0}F^{(1)}_{0}(z_1)=0$, which vanishes only for $\zeta_0=0$. Next, we divide the first equation in (\ref{lin}) by $(z_1-z_-^{(1)})$ and set $z_1=z_-^{(1)}q^{-2}$. Similarly, it implies $\zeta_{1}=0$. By induction, (\ref{lin}) follows. So the claim holds for $N=1$. We now turn to arbitrary values of $N$.
Suppose as per the statement of the Lemma that
 $z_\pm^{(i)}\neq z_\pm^{(j)}$ for any $i\neq j$. Apply the same reasoning as before.  By induction, it follows that $\{F^{(N)}_{\{n\}}(z_1,z_2,...,z_N)\}$ form a basis of $\mathcal{P}^{(N)}_z$ provided conditions (\ref{cond1zz}) are satisfied. Similarly, one shows that  $\{\tilde{F}^{(N)}_{\{\tilde{n}\}}(z_1,z_2,...,z_N)\}$ form a basis of the vector space under (\ref{cond2zz}).
\end{proof}
The action of ${\cal W}_{1}^{(i)}$ (resp. ${\cal W}_{0}^{(i)}$) on the eigenbasis of the $q-$difference operator ${\cal W}_{0}^{(i)}$ (resp. ${\cal W}_{1}^{(i)}$) is now considered. First, we describe the action of the operators for generic values of $N$.\vspace{1mm}

\begin{prop}\label{proptrid} Let $V_{i;p_i}$ (resp. $V^*_{i;\tilde{p}_i}$) denote the  subspace generated by the polynomials $F^{(N)}_{n_1 n_2...n_N}(z_1,z_2,...,z_N)$ (resp. $\tilde{F}^{(N)}_{\{\tilde{n}\}}(z_1,z_2,...,z_N)$) with fixed $p_i=n_1+n_2+...+n_i$  (resp. $\tilde{p}_i=\tilde{n}_1+\tilde{n}_2+...+\tilde{n}_i$) for any $i\in\{1,2,...,N\}$.  One has:
\beqa
\cW_{0}^{(i)}V_{i;p_i} &\subseteq& V_{i;p_i}, \label{diagW}\\
\cW_{1}^{(i)}V_{i;p_i} &\subseteq& V_{i;p_i+1} + V_{i;p_i} + V_{i;p_i-1} \qquad \qquad 0 \leq p_i \leq 2\kJ_i, \label{tridiagW}
\eeqa
where $V_{i;-1} = 0$ and $V_{i;2\kJ_i+1}= 0$.
\beqa
\cW_{1}^{(i)}V^*_{i;\tilde{p}_i} &\subseteq& V^*_{i;\tilde{p}_i}, \label{diagdualW}\\
\cW_{0}^{(i)}V^*_{i;\tilde{p}_i} &\subseteq& V^*_{i;\tilde{p}_i+1} + V^*_{i;\tilde{p}_i} + V^*_{i;\tilde{p}_i-1} \qquad \qquad 0 \leq \tilde{p}_i \leq 2\kJ_i, \label{tridiagdualW}
\eeqa
where $V^*_{i;-1} = 0$ and $V^*_{i;2\kJ_i+1}= 0$.
\end{prop}
\begin{proof} The  following arguments are essentially based on \cite[Proof of Theorem 3.10]{Ter03}. By Proposition \ref{eigenF}, (\ref{diagW}) holds. We now demonstrate (\ref{tridiagW}). By Proposition \ref{hom}, recall that $\cW_{0}^{(i)},\cW_{1}^{(i)}$ satisfy the defining relations of the $q-$Onsager algebra (\ref{qDG}). Take $E_{i;p_i}$ as the projector onto the eigenspace $V_{i;p_i}$ associated with the eigenvalue $\lambda^{(i)}_{p_i}$ of $\cW_{0}^{(i)}$. Let $\Delta$ denote the difference between the left-hand side and the  right-hand side  of the first equation in (\ref{qDG}), so that this equation reads $\Delta=0$. One has $E_{i;p_i} \Delta E_{i;m_i} = P(\lambda^{(i)}_{p_i},\lambda^{(i)}_{m_i}) \ E_{i;p_i} \cW_1^{(i)} E_{i;m_i}$ with
\beqa
P(x,y) = (x-y) (x^2 - (q^2+q^{-2}) xy +y^2 -\rho)\  .
\eeqa
For each pair of integers $p_i,m_i$ it is straightforward to show from (\ref{spec1}) that $P(\lambda^{(i)}_{p_i},\lambda^{(i)}_{m_i})=0$ if  $|p_i-m_i|\leq 1$. It implies:
\beqa
E_{i;p_i} \cW_1^{(i)} E_{i;m_i} =0  \quad \mbox{if}\quad |p_i-m_i|>1
\eeqa
which proves the claim (\ref{tridiagW}). The statements (\ref{diagdualW}),(\ref{tridiagdualW}) are shown similarly.
\end{proof}

We would like to point out some connection with the theory of tridiagonal pairs.
By \cite[Definition 2.1]{Ter01}, a tridiagonal pair of $q-$Racah type is such that (i) both operators are diagonalizable; (ii) the two operators act as (\ref{diagW}), (\ref{tridiagW}), (\ref{diagdualW}), (\ref{tridiagdualW}) on the respective eigenspaces and (iii) the vector space is irreducible. From Propositions \ref{eigenF}, \ref{eigenF2} and Lemma \ref{lem-basis}, it follows that (i) holds. By Proposition \ref{proptrid}, (ii) is then verified. If in addition we assume that the vector space is irreducible,  for any $i=0,1,...,N$ it follows that the $q-$difference operators $\cW_{0}^{(i)}$ and $\cW_{1}^{(i)}$ form a tridiagonal pair of $q-$Racah type. 
\vspace{2mm}

To illustrate Proposition  \ref{proptrid},
we shall describe below the cases $N=1$ and $N=2$ in some details.
\begin{example}\label{ex1}
 The $q-$difference operator $\cW_1^{(1)}$  acts  on the  polynomial eigenfunction $F^{(1)}_{n_1}(z_1)$ given by (\ref{FN}) as follows:
\beqa
\cW_1^{(1)}F^{(1)}_{n_1}(z_1)={\cal B}_{n_1}^{[1]}F^{(1)}_{n_1+1}(z_1) + {\cal C}_{n_1}^{[-1]} F^{(1)}_{n_1-1}(z_1) + {\cal A}_{n_1}^{[0]}F^{(1)}_{n_1}(z_1),\label{recW1}
\eeqa
where
\beqa
{\cal B}_{n_1}^{[1]} &=& -b_1^{(1)}q^{2j_1} \tilde{z}_+^{(1)} \frac{(1-q^{2(1+n_1-2j_1)}z_+^{(1)}/\tilde{z}_-^{(1)})}{(1-q^{-2n_1}z_-^{(1)}/\tilde{z}_-^{(1)})}
B_{n_1}(q^{-{4j_1}}z_+^{(1)}/\tilde{z}_+^{(1)},q^{-2}\tilde{z}_+^{(1)}/z_-^{(1)},q^{-{4j_1}-2},\tilde{z}_-^{(1)}/\tilde{z}_+^{(1)}),\label{ABC}\nonumber\\
{\cal C}_{n_1}^{[-1]} &=& -b_1^{(1)}q^{2j_1} \tilde{z}_+^{(1)}  \frac{(1-q^{2(1-n_1)}z_-^{(1)}/\tilde{z}_-^{(1)})}{(1-q^{2(n_1-{2j_1})}z_+^{(1)}/\tilde{z}_-^{(1)})} 
C_{n_1}(q^{-{4j_1}}z_+^{(1)}/\tilde{z}_+^{(1)},q^{-2}\tilde{z}_+^{(1)}/z_-^{(1)},q^{-{4j_1}-2},\tilde{z}_-^{(1)}/\tilde{z}_+^{(1)}),\nonumber\\
{\cal A}_{n_1}^{[0]} &=& -b_1^{(1)}q^{2j_1} \tilde{z}_+^{(1)}\left(1+q^{-{4j_1}}\tilde{z}_-^{(1)}/\tilde{z}_+^{(1)} - B_{n_1}-C_{n_1}\right)\nonumber
\eeqa
with $b_1^{(1)},B_{n_1},C_{n_1}$ respectively given in (\ref{coeffbc}), (\ref{Bcoeff}), (\ref{Ccoeff}).
\end{example}
\begin{proof}  Using (\ref{FN}) for $N=1$, eq. (\ref{recW1}) explicitly reads:
\beqa
\qquad \left( b_1^{(1)} q^{2j_1} z_1 + c_1^{(1)} z_1^{-1} + \epsilon_-q^{2j_1} \right)f^{(1)}_{n_1}(q^{-2}z_1) - \left(b_1^{(1)} q^{-2j_1}z_1 +   c_1^{(1)}z_1^{-1}\right)f^{(1)}_{n_1}(z_1)&=& \label{w1act}\\ 
{\cal B}_{n_1}^{[1]}f^{(1)}_{n_1+1}(z_1) + {\cal C}_{n_1}^{[-1]} f^{(1)}_{n_1-1}(z_1)\!\!\!\! &+& \!\!\!\!\! {\cal A}_{n_1}^{[0]}f^{(1)}_{n_1}(z_1).\nonumber
\eeqa
Observe that:
\beqa
f^{(1)}_{n_1}(q^{-2}z_1) &=& q^{-4j_1}(z_1 - z_{+}^{(1)}q^{2})(z_1 - z_{-}^{(1)}q^{2})Q_{n_1}(z_1),\label{decomp}\\
f^{(1)}_{n_1}(z_1) &=& (z_1 - z_{+}^{(1)}q^{-4j_1+2n_1+2})(z_1 - z_{-}^{(1)}q^{-2n_1+2})Q_{n_1}(z_1),\nonumber\\
f^{(1)}_{n_1+1}(z_1) &=& (z_1 - z_{-}^{(1)}q^{-2n_1})(z_1 - z_{-}^{(1)}q^{-2n_1+2})Q_{n_1}(z_1),\nonumber\\
f^{(1)}_{n_1-1}(z_1) &=& (z_1 - z_{+}^{(1)}q^{-4j_1+2n_1})(z_1 - z_{+}^{(1)}q^{-4j_1+2n_1+2})Q_{n_1}(z_1),\nonumber
\eeqa
where
\beqa
Q_{n_1}(z_1)=\prod_{k=0}^{2j_1-n_1-2}   (z_1 - z_{+}^{(1)}q^{-2k})   \prod_{l=0}^{n_1-2} (z_1 - z_{-}^{(1)}q^{-2l}).\nonumber
\eeqa
Inserting (\ref{decomp}) into (\ref{w1act}) and dividing by $Q_{n_1}(z_1)$, one obtains a polynomial equation\footnote{The coefficients of $z_1^3$ and $z_1^{-1}$ are vanishing.} of degree two in $z_1$. Upon requiring that the coefficients vanish, one obtains three equations that determine uniquely  ${\cal A}_{n_1}^{[1]},{\cal B}_{n_1}^{[-1]},{\cal C}_{n_1}^{[0]}$. 
\end{proof}

\begin{example}\label{ex2}
The mutually commuting $q-$difference operators $\cW_1^{(1)}$ and $\cW_1^{(2)}$ act on the  polynomial eigenfunction $F^{(2)}_{n_1n_2}(z_1,z_2)$ given by (\ref{FN}) as follows:
\beqa
\quad \cW_1^{(1)}F^{(2)}_{n_1n_2}(z_1,z_2)={\cal B}_{n_1}^{[1]}F^{(1)}_{n_1+1}(z_1)f^{(2)}_{n_2}(z_2) + {\cal C}_{n_1}^{[-1]} F^{(1)}_{n_1-1}(z_1)f^{(2)}_{n_2}(z_2) + {\cal A}_{n_1}^{[0]} F^{(2)}_{n_1 n_2}(z_1,z_2),\label{recW12}
\eeqa
and
\beqa
\cW_1^{(2)}F^{(2)}_{n_1n_2}(z_1,z_2) &=&  {\cal B}^{[10]}_{n_1n_2} F^{(2)}_{n_1+1n_2}(z_1,z_2) +  {\cal B}^{[01]}_{n_1n_2}F^{(2)}_{n_1n_2+1}(z_1,z_2) +  {\cal B}^{[-12]}_{n_1n_2}F^{(2)}_{n_1-1n_2+2}(z_1,z_2) \nonumber \\
&&+ \   {\cal C}^{[-10]}_{n_1n_2} F^{(2)}_{n_1-1n_2}(z_1,z_2) +  {\cal C}^{[0-1]}_{n_1n_2}F^{(2)}_{n_1n_2-1}(z_1,z_2) +  {\cal C}^{[1-2]}_{n_1n_2}F^{(2)}_{n_1+1n_2-2}(z_1,z_2)\label{recW122} \\
&&+  \  {\cal A}^{[1-1]}_{n_1n_2} F^{(2)}_{n_1+1n_2-1}(z_1,z_2) +   {\cal A}^{[-11]}_{n_1n_2} F^{(2)}_{n_1-1n_2+1}(z_1,z_2)  +   {\cal A}^{[00]}_{n_1n_2} F^{(2)}_{n_1n_2}(z_1,z_2),\nonumber
\eeqa
where
\beqa
{\cal B}^{[10]}_{n_1n_2}= q^{2j_2}B_2 {\cal B}^{[1]}_{n_1}, \quad {\cal B}^{[01]}_{n_1n_2}= B_{aux}, \quad {\cal B}^{[-12]}_{n_1n_2}= q^{2j_2}B'_2 {\cal C}^{[1]}_{n_1} ,\nonumber\\
{\cal C}^{[-10]}_{n_1n_2}= q^{2j_2}C'_2 {\cal C}^{[1]}_{n_1}, \quad {\cal C}^{[0-1]}_{n_1n_2}= C_{aux}, \quad {\cal C}^{[1-2]}_{n_1n_2}= q^{2j_2}C_2 {\cal B}^{[1]}_{n_1},\nonumber\\
{\cal A}^{[1-1]}_{n_1n_2}= q^{2j_2}A_2 {\cal B}^{[1]}_{n_1} , \quad {\cal A}^{[-11]}_{n_1n_2}= q^{2j_2}A'_2 {\cal C}^{[1]}_{n_1}, \quad {\cal A}^{[00]}_{n_1n_2}= A_{aux}  ,\nonumber
\eeqa
with $A_2,B_2,C_2,A'_2,B'_2,C'_2,A_{aux},B_{aux},C_{aux}$ given in Appendix B.
\end{example}
\begin{proof} The result (\ref{recW12}) is immediate from Example \ref{ex1}. We now prove (\ref{recW122}). According to the tensor product structure (\ref{deltadef}), the $q-$difference operator ${\cW}_{1}^{(2)}$ can be written as:
\beqa
{\cal W}_{1}^{(2)}&=&  b_1^{(2)} z_2 \left(q^{2j_2}{T_-^{(2)}}^2 -q^{-2j_2}\right)+ c_1^{(2)}  z_2^{-1}\left({T_-^{(2)}}^2-1\right) +  q^{2j_2} {T_-^{(2)}}^2 {\cal W}_{1}^{(1)}.\label{qdiffW12}
\eeqa
Acting on $F^{(2)}_{n_1n_2}(z_1,z_2)= f^{(2)}_{n_2}(z_2)f^{(1)}_{n_1}(z_1)$ and using (\ref{recW1}), it follows:
\beqa
\cW_1^{(2)}F^{(2)}_{n_1n_2}(z_1,z_2) &=& \left(\left( b_1^{(2)} q^{2j_2} z_2 + c_1^{(2)} z_2^{-1} + q^{2j_2}{\cal A}_{n_1}^{(1)} \right)f^{(2)}_{n_2}(q^{-2}z_2)\right. \label{eqW2} \\
&& \qquad \qquad - \left.\left(b_1^{(2)} q^{-2j_2}z_2 +   c_1^{(2)}z_2^{-1}\right)f^{(2)}_{n_2}(z_2)\right)f^{(1)}_{n_1}(z_1)\nonumber\\
&& \qquad \qquad + \  q^{2j_2}{\cal B}_{n_1}^{(1)}f^{(2)}_{n_2}(q^{-2}z_2)f^{(1)}_{n_1+1}(z_1)\nonumber\\
&& \qquad \qquad + \ q^{2j_2}{\cal C}_{n_1}^{(1)}f^{(2)}_{n_2}(q^{-2}z_2)f^{(1)}_{n_1-1}(z_1).\nonumber
\eeqa
Consider the combination $f^{(2)}_{n_2}(q^{-2}z_2)f^{(1)}_{n_1}(z_1)$. In analogy with the analysis for the case $N=1$, a straightforward computation allows to identify $A_{aux},B_{aux},C_{aux}$ such that:
\beqa
\qquad \quad\left(b_1^{(2)} q^{2j_2} z_2 + c_1^{(2)} z_2^{-1} +q^{2j_2}{\cal A}_{n_1}^{(1)} \right)f^{(2)}_{n_2}(q^{-2}z_2) - \left(b_1^{(2)} q^{-2j_2}z_2 +   c_1^{(2)}z_2^{-1}\right)f^{(2)}_{n_2}(z_2)&=& \label{w12aux1}\\ 
B_{aux}f^{(2)}_{n_2+1}(z_2) + C_{aux}f^{(2)}_{n_2-1}(z_2) \!\!\!\!&+&\!\!\!\! A_{aux}f^{(2)}_{n_2}(z_2).\nonumber
\eeqa
The expressions for $A_{aux},B_{aux},C_{aux}$ are given in Appendix B. Now, consider the combinations $f^{(2)}_{n_2}(q^{-2}z_2)f^{(1)}_{n_1+1}(z_1)$ and $f^{(2)}_{n_2}(q^{-2}z_2)f^{(1)}_{n_1-1}(z_1)$. Recall that $z_\pm^{(2)}$ depends on $n_1$. According to (\ref{zpmi}), one has: 
\beqa
z_\pm^{(2)}|_{n_1\rightarrow n_1+1} = q^{\pm 2}z_\pm^{(2)} \quad \mbox{and} \quad z_\pm^{(2)}|_{n_1\rightarrow n_1-1} = q^{\mp 2}z_\pm^{(2)}.\nonumber
\eeqa
As a consequence, the following two relations
\beqa
f^{(2)}_{n_2}(q^{-2}z_2)&=& \label{w12aux2}
B_2f^{(2)}_{n_2}(z_2)|_{n_1\rightarrow n_1+1} + C_2f^{(2)}_{n_2-2}(z_2)|_{n_1\rightarrow n_1+1}+ A_2f^{(2)}_{n_2-1}(z_2)|_{n_1\rightarrow n_1+1},\\
f^{(2)}_{n_2}(q^{-2}z_2)&=& \label{w12aux3}
B'_2f^{(2)}_{n_2+2}(z_2)|_{n_1\rightarrow n_1-1} + C'_2f^{(2)}_{n_2}(z_2)|_{n_1\rightarrow n_1-1}+ A'_2f^{(2)}_{n_2+1}(z_2)|_{n_1\rightarrow n_1-1},
\eeqa
determine uniquely $A_2,B_2,C_2$ and $A'_2,B'_2,C'_2$. Their formulas are also given in Appendix B. Finally, observe that:
\beqa
F^{(2)}_{n_1\pm 1 n_2+a}(z_1,z_2)= f^{(2)}_{n_2+a}(z_2)|_{n_1\rightarrow n_1\pm 1}f^{(1)}_{n_1\pm 1}(z_1) \quad \mbox{for} \quad a=0,\mp 1,\mp 2.
\eeqa
Inserting (\ref{w12aux1}), (\ref{w12aux2}), (\ref{w12aux3})  into (\ref{eqW2}) and combining all coefficients, one ends up with  (\ref{recW122}).  
\end{proof}

\vspace{1mm}

The previous examples above make clear that the analysis for an arbitrary $N$ can be achieved by induction in a straightforward manner. Note that the special case $j_1=j_2=...=j_N=1/2$ is treated in details in \cite{Bas3}.

\section{Overlap coefficients and orthogonality}
In the previous section, we have constructed two different bases of the polynomial vector space $\mathcal{P}^{(N)}_z$. By Proposition \ref{proptrid}, for any $i=1,2,...,N$, in the first (resp. second) basis,  the operator $\cW_0^{(i)}$ (resp. $\cW_1^{(i)}$) is a  diagonal matrix whereas the operator $\cW_1^{(i)}$ (resp. $\cW_0^{(i)}$) is a block tridiagonal matrix. In this section,
the overlap coefficients between the two bases are studied. They are derived in terms of an entangled product of $q-$Racah orthogonal polynomials and shown to satisfy certain orthogonality relations.
\begin{lem}\label{Cf}  Let $f_{n_i}^{(i)}(z_i), i=1,2,...,N$,  be the functions  defined in (\ref{FN}) and take $\tilde{f}_{\tilde{n}_i}^{(i)}(z_i)=$\\
$ f_{n_i}^{(i)}(z_i)|_{q\rightarrow q^{-1},n_i\rightarrow, \tilde{n}_iz_\pm^{(i)}\rightarrow \tilde{z}_\pm^{(i)}}$. The following expansion formulas hold:
\beqa
 f_{n_i}^{(i)}(z_i) =   \sum_{\tilde{n}_i=0}^{2j_i} C_{n_i}^{\tilde{n}_i}(z_-^{(i)},z_+^{(i)},\tilde{z}_-^{(i)},\tilde{z}_+^{(i)};2j_i;q^2)   \tilde{f}_{\tilde{n}_i}^{(i)}(z_i)  \quad  \mbox{for} \quad n_i=0,1,...,2j_i\label{basicexp}
\eeqa
where
\beqa
\qquad  C_{n}^{\tilde{n}}(a,b,c,d;M;q^2) \!\!\!\!\!&=&\!\!\!\!\! q^{2\tilde{n}(\tilde{n}-M)}  \left[ \begin{array}{c} M \\   \tilde{n} \end{array}\right]_{q^2}  \frac{ (q^{2(1-M)}b/d;q^2)_{\tilde{n}}    (q^{2(1-M)}b/c;q^2)_{M-\tilde{n}}   (q^{2(1-n)}a/c;q^2)_{n}   }{    (q^{2(\tilde{n}-M)}c/d;q^2)_{\tilde{n}}    (q^{-2\tilde{n}}d/c;q^2)_{M-\tilde{n}}   (q^{2(1-M)}b/c;q^2)_{n}   }\label{Cnn}\\
&&\qquad \qquad \qquad  \quad \times  \ R_n(\mu(\tilde{n});q^{-2M}b/d,q^{-2}d/a,q^{-2M-2},c/d;q^2). \nonumber
\eeqa
%
\end{lem}
\begin{proof} Write  $ f_{n_i}^{(i)}(z_i)$ and $\tilde{f}_{\tilde{n}_i}^{(i)}(z_i)$ using (\ref{FN}) in terms of $q-$Pochhammer functions. Apply the expansion formula \cite[eq. (2.16)]{Ro2} with $q\rightarrow q^2$:
\beqa
(ax;q^{-2})_n (bx;q^{-2})_{N-n} =\sum_{\tilde{n}=0}^{N} C_{n}^{\tilde{n}}(a,b,c,d;N;q^2) (cx;q^{2})_{\tilde{n}} (dx;q^{2})_{N-\tilde{n}}.\label{expC}
\eeqa
\end{proof}

\begin{rem}\label{rem41} The three-term recurrence relations (\ref{threerec1}) satisfied by the $q-$Racah polynomials (\ref{qR}) can be recovered as follows. Let us introduce the auxiliary operator 
\beqa
w_{1}^{(i)}&=&  b_1^{(i)} z_i \left(q^{2j_i}{T_-^{(i)}}^2 -q^{-2j_i}\right)+ c_1^{(i)}  z_i^{-1}\left({T_-^{(i)}}^2-1\right) +  \tilde{\lambda}_{\tilde{n}_1,\tilde{n}_2,...,\tilde{n}_{i-1}}^{(i-1)} q^{2j_i}{T_-^{(i)}}^2.\label{qdiffW1i}
\eeqa
On the one hand, for any  $n_i=0,1,...,2j_i$ from (\ref{basicexp}) one has:
\beqa
\quad w_{1}^{(i)} f_{n_i}^{(i)}(z_i) =   \sum_{\tilde{n}_i=0}^{2j_i} \tilde{\lambda}_{\tilde{n}_1,\tilde{n}_2,...,\tilde{n}_{i}}^{(i)} C_{n_i}^{\tilde{n}_i}(z_-^{(i)},z_+^{(i)},\tilde{z}_-^{(i)},\tilde{z}_+^{(i)};2j_i;q^2)   \tilde{f}_{\tilde{n}_i}^{(i)}(z_i). \label{actaux}
\eeqa
On the other hand, generalizing the arguments used in the proof of Example \ref{ex1}, one shows that:
\beqa
w_1^{(i)}f_{n_i}^{(i)}(z_i)={\cal B}_{n_i}^{[1]}f_{n_i+1}^{(i)}(z_i) + {\cal C}_{n_i}^{[-1]} f_{n_i-1}^{(i)}(z_i) + {\cal A}_{n_i}^{[0]}f_{n_i}^{(i)}(z_i).\label{recw1i}
\eeqa
where the coefficients ${\cal B}_{n_i}^{[1]},{\cal C}_{n_i}^{[-1]}$ and ${\cal A}_{n_i}^{[0]}$ are obtained from ${\cal B}_{n_1}^{[1]},{\cal C}_{n_1}^{[-1]}$ and ${\cal A}_{n_1}^{[0]}$ in (\ref{recW1}) through the substitutions:
\beqa
n_1\rightarrow n_i,\quad z_1\rightarrow z_i,\quad  j_1\rightarrow j_i,\quad b_1^{(1)}  \rightarrow b_1^{(i)},\quad   z_\pm^{(1)}\rightarrow z_\pm^{(i)},\quad  \tilde{z}_\pm^{(1)} \rightarrow \tilde{z}_\pm^{(i)}.\nonumber 
\eeqa
Inserting (\ref{basicexp}) into the r.h.s of (\ref{recw1i}) and equating the resulting expression with the r.h.s of (\ref{actaux}), one ends up with a three-term recurrence relation on the coefficients (\ref{Cnn}). After some simplifications, one obtains the relation (\ref{threerec1}).  
\end{rem}

Recall the definitions of $F^{(N)}_{\{n\}}(z_1,z_2,...,z_N)$ and   $\tilde{F}^{(N)}_{\{\tilde{n}\}}(z_1,z_2,...,z_N)$  in  Proposition \ref{eigenF} and Proposition \ref{eigenF2}, respectively.  By Lemma \ref{Cf}:
\begin{prop}\label{thm3}  The following expansion formulas hold: 
\beqa
F^{(N)}_{\{n\}}(\{z\})   =    \sum_{\{\tilde{n}\}=\{0\}^N}^{\{2j\}} C_{\{n\}}^{\{\tilde{n}\}}\left(\{z_-\},\{z_+\},\{\tilde{z}_-\},\{\tilde{z}_+\};\{2j\};q^2\right)  \tilde{F}^{(N)}_{\{\tilde{n}\}}(\{z\}) \label{expN}
\eeqa
where 
\beqa
C_{\{n\}}^{\{\tilde{n}\}}\left(\{z_-\},\{z_+\},\{\tilde{z}_-\},\{\tilde{z}_+\};\{2j\};q^2\right) = \prod_{i=1}^{N} C_{n_i}^{\tilde{n}_i}(z_-^{(i)},z_+^{(i)},\tilde{z}_-^{(i)},\tilde{z}_+^{(i)};2j_i;q^2) . \label{prodC}
\eeqa
\end{prop}

Applying twice the expansion formulas (\ref{expC}), one finds that the basic overlap coefficients (\ref{Cnn}) satisfy an orthogonality relation. It is easy to show 
that this orthogonality relation is a consequence 
of the well-known orthogonality relation (\ref{orthoqRac}) of the $q-$Racah polynomials. For generic values of $N$, we extend the argument. It follows that the overlap coefficients (\ref{prodC}) satisfy a generalized orthogonality relation. 
\begin{prop} The following orthogonality relation holds:
\beqa
\sum_{\{\tilde{n}\}=\{0\}^N}^{\{2j\}} C_{\{n\}}^{\{\tilde{n}\}}\left(\{z_-\},\{z_+\},\{\tilde{z}_-\},\{\tilde{z}_+\};\{2j\};q^2\right) C_{\{\tilde{n}\}}^{\{m\}}\left(\{\tilde{z}_-\},\{\tilde{z}_+\},\{z_-\},\{z_+\},;\{2j\};q^{-2}\right) = \delta_{\{n\},\{m\}}.\nonumber
\eeqa
\end{prop}

\vspace{2mm}

In the case when $N=2$ it is straightforward to observe that the coefficients (\ref{prodC}) obey a $3-$term and a $9-$ term recurrence relation that are obtained from (\ref{recW12}) and (\ref{recW122}). One then note that these recurrence relations have a structure analogous to that of the bivariate Hahn polynomials and Gasper--Rahman polynomials found in \cite{GV} and \cite{BM}, respectively. This suggests that the coefficients $C_{\{n\}}^{\{\tilde{n}\}}$ could be related to known multivariate $q-$polynomials. This will be further explored in a separate study.\vspace{1mm}

To conclude this Section, we now consider the overlap coefficients between the eigenfunctions of $\cW_0^{(i)},\cW_1^{(i)}$ previously constructed and the multivariable monomial basis. In this case, the overlap coefficients  are identified as multivariate extensions of the univariate dual $q-$Krawtchouk polynomials \cite[Section 3.17]{KS}.
\begin{prop}\label{prop5} Let $\tilde{F}^{(N)}_{\{\tilde{n}\}}(z_1,z_2,...,z_N)$  be defined in Proposition  \ref{eigenF2}. The following expansion formulas hold: 
\beqa
\tilde{F}^{(N)}_{\{\tilde{n}\}}(\{z\})   =    \sum_{\{n\}=\{0\}^N}^{\{2j\}} 
\tilde{D}_{\{\tilde{n}\}}^{\{n\}}\left(\{\tilde{z}_-\},\{\tilde{z}_+\};\{2j\};q^2\right) 
  z_1^{n_1}z_2^{n_2}\cdots z_N^{n_N},\label{Ftildemon}
\eeqa
where 
\beqa
 \tilde{D}_{\{\tilde{n}\}}^{\{n\}}\left(\{\tilde{z}_-\},\{\tilde{z}_+\};\{2j\};q^2\right)  = \prod_{i=1}^N  (q^{4j_i}\tilde{z}^{(i)}_+)^{2j_i-n_i}\frac{(q^{-4j_i};q^2)_{2j_i-n_i}}{(q^2;q^2)_{2j_i-n_i}} K_{2j_i-n_i}  \left(   \mu(\tilde{n}_i);\tilde{z}_-^{(i)}/\tilde{z}_+^{(i)},2j_i;q^2\right)  \label{prodDtilde}\nonumber
\eeqa
with (\ref{qK}).
\end{prop}
\begin{proof}  Rewrite  $\tilde{F}^{(N)}_{\{\tilde{n}\}}(\{z\})$ in terms of $q-$Pochhammer functions.  Use the expansion formula (3.17.11) of \cite{KS}.
\end{proof}

\begin{rem} By analogy with Remark \ref{rem41}, the three-term recurrence relations given in Appendix \ref{B3} (see \cite[eq. 3.17.3]{KS}) satisfied by the  dual $q-$Krawtchouk polynomial (\ref{qK}) can be obtained using the action of the auxiliary operator (\ref{qdiffW1i}) on (\ref{Ftildemon}).
\end{rem}

The overlap coefficients between $F^{(N)}_{\{n\}}(z_1,z_2,...,z_N)$   and the multivariable monomial basis are similarly found:
 
\begin{prop}
Let $F^{(N)}_{\{n\}}(z_1,z_2,...,z_N)$  be defined in Proposition  \ref{eigenF}. The following expansion formulas hold: 
\beqa
F^{(N)}_{\{n\}}(\{z\})   =  \prod_{i=1}^{N}(-1)^{2j_i}(z^{(i)}_+)^{2j_i-n_i}(z^{(i)}_-)^{n_i} q^{-2(_{ 2}^{n_i})-2(_{\ \ 2}^{2j_i-n_i})}  \sum_{\{\tilde{n}\}=\{0\}^N}^{\{2j\}} 
D_{\{n\}}^{\{\tilde{n}\}}\left(\{z_-\},\{z_+\};\{2j\};q^2\right) 
  z_1^{\tilde{n}_1}z_2^{\tilde{n}_2}\cdots z_N^{\tilde{n}_N},\label{Fmon}\nonumber
\eeqa
where 
\beqa
 D_{\{n\}}^{\{\tilde{n}\}}\left(\{z_-\},\{z_+\};\{2j\};q^2\right)  = \prod_{i=1}^N  \left(\frac{q^{4j_i}}{z^{(i)}_+}\right)^{\tilde{n}_i}\frac{(q^{-4j_i};q^2)_{\tilde{n}_i}}{(q^2;q^2)_{\tilde{n}_i}} K_{\tilde{n}_i}  \left(   \mu(n_i);z_+^{(i)}/z_-^{(i)},2j_i;q^2\right).  \label{prodD}\nonumber
\eeqa
\end{prop} 

\vspace{2mm}

\section{The `split' basis}
In this section, we introduce another basis for the polynomial vector space $\mathcal{P}^{(N)}_z$. This basis interpolates between the two eigenbases constructed in the previous section, and can be understood as a generalization of the `split'  (one-variable)  basis proposed in \cite[Remark 2.1]{Ro2}. With respect to this basis, it is shown that the $q-$difference operators $\cW_0^{(i)},\cW_1^{(i)}$ of Proposition \ref{realqdiff} act as upper and lower block bidiagonal matrices, respectively.
Define
\beqa
 \qquad G^{(N)}_{\{n\}}(\{z\}) &=& \prod_{i=1}^{N} g_{n_i}^{(i)}(z_i)  \quad \mbox{with} \quad
 g_{n_i}^{(i)}(z_i)=\prod_{k=0}^{n_i-1}   (z_i - z_{-}^{(i)}q^{-2k})   \prod_{l=0}^{2j_i-1-n_i} (z_i - \overline{z}_+^{(i)}q^{2l})\label{GN}
\eeqa
where
\beqa \overline{z}_+^{(i)}=\tilde{z}_{+}^{(i)}|_{\{\tilde{n}\}\rightarrow \{n\}}.\label{zbar}
\eeqa
For the proof of the following Lemma, we proceed by analogy with the derivation of Lemma \ref{lem-basis}.
\begin{lem}\label{lem-basis2} Assume
\beqa
&&\overline{z}_+^{(i)}/z_-^{(i)}\notin\{1,q^{-2},...,q^{-4j_i+2}\}, \quad z_-^{(i)}\neq 0,  \quad \overline{z}_+^{(i)}\neq 0,\quad z_-^{(i)}\neq z_-^{(j)}, \quad \overline{z}_+^{(i)}\neq \overline{z}_+^{(j)}\label{cond3zz}
\eeqa
for any $i\neq j$. The $N-$variable polynomial vector space $\mathcal{P}^{(N)}_z$ admits a basis generated by the polynomials $\{G^{(N)}_{\{n\}}(z_1,z_2,...,z_N)\}$. The cardinality is given by $\prod_{k=1}^N(2j_k+1)$.
\end{lem}

Let us consider the cases $N=1$ and $N=2$. The proof of the following results essentially follows that of  Proposition \ref{ex1}. The details are omitted. 
\begin{example}\label{ex3}
On the  polynomial $G^{(1)}_{n_1}(z_1)$ given by (\ref{GN}), the $q-$difference operator $\cW_0^{(1)}$, $\cW_1^{(1)}$  act, respectively, as:
\beqa
\cW_0^{(1)}G^{(1)}_{n_1}(z_1)&=&{\cal D}_{n_1}^{[1]}G^{(1)}_{n_1+1}(z_1) + \lambda_{n_1}^{(1)} G^{(1)}_{n_1}(z_1),\label{GbirecW0}\\
\cW_1^{(1)}G^{(1)}_{n_1}(z_1)&=&{\cal E}_{n_1}^{[-1]}G^{(1)}_{n_1-1}(z_1) + \tilde{\lambda}_{n_1}^{(1)}  G^{(1)}_{n_1}(z_1)\nonumber
\eeqa
where
\beqa
{\cal D}_{n_1}^{[1]}&=&b_0^{(1)}(1-q^{4j_1-2n_1})\left(\overline{z}_+^{(1)} q^{2j_1-2} -z_+^{(1)} q^{2n_1-2j_1}\right),
\nonumber\\
{\cal E}_{n_1}^{[-1]} &=& b_1^{(1)}(q^{2n_1}-1)\left(\overline{z}_-^{(1)} q^{-2j_1} -z_-^{(1)} q^{2-2j_1-2n_1}\right)\nonumber.
\nonumber
\eeqa
%
%
%
\end{example}

\begin{example}
On the  polynomial $G^{(2)}_{n_1n_2}(z_1,z_2)$ given by (\ref{GN}), the $q-$difference operators $\cW_0^{(i)}$, $\cW_1^{(i)}$, $i=1,2$  act, respectively, as:
\beqa
\qquad \cW_0^{(2)}G^{(2)}_{n_1n_2}(z_1,z_2)&=&{\cal D}_{n_1n_2}^{[10]}G^{(2)}_{n_1+1n_2}(z_1,z_2) + {\cal D}_{n_1n_2}^{[01]}G^{(2)}_{n_1n_2+1}(z_1,z_2)    +  \lambda_{n_1,n_2}^{(2)}  G^{(2)}_{n_1n_2}(z_1,z_2) ,\label{GbirecW02}\\
\cW_0^{(1)}G^{(2)}_{n_1n_2}(z_1,z_2) &=&{\cal D}_{n_1}^{[1]}G^{(1)}_{n_1+1}(z_1)g^{(2)}_{n_2}(z_2)  + \lambda_{n_1}^{(1)}  G^{(2)}_{n_1n_2}(z_1,z_2) \nonumber
\eeqa
where
\beqa
{\cal D}_{n_1n_2}^{[10]}= q^{2j_2}{\cal D}_{n_1}^{[1]}   \qquad \mbox{and}\qquad  {\cal D}_{n_1n_2}^{[01]}=  {\cal D}_{n_1}^{[1]}|_{(n_1,j_1, b_0^{(1)},z_+^{(1)}, \overline{z}_+^{(1)}) \rightarrow (n_2, j_2, b_0^{(2)}, z_+^{(2)}, \overline{z}_+^{(2)})} \nonumber
\eeqa
and
\beqa
\cW_1^{(2)}G^{(2)}_{n_1n_2}(z_1,z_2)&=&{\cal E}_{n_1n_2}^{[-10]}G^{(2)}_{n_1-1n_2}(z_1,z_2) + {\cal E}_{n_1n_2}^{[0-1]}G^{(2)}_{n_1n_2-1}(z_1,z_2)    +   \tilde{\lambda}_{n_1,n_2}^{(2)}  G^{(2)}_{n_1n_2}(z_1,z_2) ,
\nonumber\\
\cW_1^{(1)}G^{(2)}_{n_1n_2}(z_1,z_2) &=&{\cal E}_{n_1}^{[1]}G^{(1)}_{n_1-1}(z_1)g^{(2)}_{n_2}(z_2)  + \tilde{\lambda}_{n_1}^{(1)}  G^{(2)}_{n_1n_2}(z_1,z_2) \nonumber
\eeqa
where
\beqa
{\cal E}_{n_1n_2}^{[-10]}= q^{-2j_2}{\cal E}_{n_1}^{[-1]}   \qquad \mbox{and}\qquad  {\cal E}_{n_1n_2}^{[0-1]}=  {\cal E}_{n_1}^{[-1]}|_{(n_1,j_1, b_1^{(1)},z_-^{(1)}, \overline{z}_-^{(1)}) \rightarrow (n_2, j_2, b_1^{(2)}, z_-^{(2)}, \overline{z}_-^{(2)})}. \nonumber
\eeqa
\end{example}

\vspace{2mm}

%
%

The analysis extends to generic values of $N$ in view of the structure of the $q-$difference operators $\cW_0^{(i)}$, $\cW_1^{(i)}$, $i=1,...,N$.  Using induction, the following proposition is straightforwardly derived.
\begin{prop} Let $U_{i;p_i}$ denotes the  subspace generated by the polynomials $G^{(N)}_{n_1 n_2...n_N}(z_1,z_2,...,z_N)$ with fixed $p_i=n_1+n_2+...+n_i$  for any $i\in\{1,2,...,N\}$.  One has:
\beqa
\left(\cW_{0}^{(i)}- \lambda_{\{n\}}^{(i)}\right)  U_{i;p_i} &\subseteq& U_{i;p_i+1}, \label{bidiagW}\\
\left(\cW_{1}^{(i)}- \tilde{\lambda}_{\{n\}}^{(i)}\right) U_{i;p_i}  &\subseteq& U_{i;p_i-1}, \qquad \qquad 0 \leq p_i \leq 2\kJ_i, \label{bidiag2W}\nonumber
\eeqa
where $U_{i;-1} = 0$ and $U_{i;2\kJ_i+1}= 0$.
\end{prop}

\section{Concluding remarks}
The results presented here open several perspectives. 
We certainly intend to develop the characterization of the multivariate special functions (\ref{prodC}) that have arisen in our study with an eye to their potential polynomiality. The following three directions also seem promising.\vspace{1mm}

Firstly, higher rank generalizations of the
$q-$Onsager algebra, denoted $O_q(\widehat{g})$, have been introduced \cite[Definition 2.1]{BB}. In analogy with the $sl_2$ case, they can be understood as certain
coideal subalgebras of $U_q(\widehat{g})$ for any affine Lie algebra $\widehat{g}$, see \cite[Proposition 2.1]{BB} (see also \cite{Ko}).
In view of the results presented here, an interesting problem would
be to construct multivariable $q-$difference operators  for the basic generators of $O_q(\widehat{g})$   and their respective polynomial eigenfunctions, expressed as entangled products of $q-$Pochhammer functions generalizing \cite{Ro2}. Such expressions should find applications in the context of quantum integrable models associated with higher rank symmetries, and provide a $q-$hypergeometric formulation of these models. For the $sl_2$ case, an example of such description is given in \cite{BM}.\vspace{1mm}

Secondly, as recently indicated, the representation theory of the so-called {\it asymetric} tridiagonal algebra \cite[Subsection 5.4]{BGV} is, in the simplest case, connected with the representation theory of the complementary Bannai--Ito and dual $-1$ Hahn algebras \cite{GVZ2}, as with the construction of univariate orthogonal polynomials beyond\footnote{They satisfy a bispectral problem associated with a three-term recurrence  relation and a five-term difference equation \cite{GVZ2}.} the Askey-scheme. The asymetric tridiagonal algebra is closely related with the $q-$Onsager 
algebra specialized to $q$ a root of unity
 (for details, see \cite{BGV}). It would thus be of interest to study multivariate generalizations of these polynomials based on the coideal structure of the $q-$Onsager algebra for $q$ taken to be roots of unity, along the approach presented here.\vspace{1mm}

Thirdly, besides the infinite dimensional representations of the $q-$Onsager algebra built from the $q-$vertex operators formalism of $U_q(\widehat{sl_2})$ (see \cite{BB3}), it would be interesting to extend the analysis of the present paper to the limit $N\rightarrow \infty$. This should find applications in the context of quantum integrable models in the analysis of the thermodynamic limit of spin chains with open boundaries.\vspace{1mm}
 
Some of these problems will be adressed elsewhere.

\vspace{4mm}

\noindent{\bf Acknowledgements:} We would like to thank V. X. Genest for discussions, X. Martin and P. Terwilliger for comments on the manuscript. P.B. is grateful to  H. Rosengren for communications. We thank N. Cramp\'e for pointing out some typos in published version, that are corrected in this version. 
P.B. also acknowledges the hospitality and support from the Centre de Recherches Math\'ematiques and CRM-UMI 3457 C.N.R.S. where most of this work has been done.
 P.B. is supported by C.N.R.S. The work of LV is funded by a grant from the Natural Sciences and Engineering Research Council (NSERC) of Canada.

\vspace{0.2cm}

\begin{appendix}
\section{$U_{q}(\widehat{sl_2})$, $U_{q}(sl_2)$ and polynomial bases}
\subsection{The quantum algebra  $U_{q}(\widehat{sl_2})$}\label{A1}

The quantum Kac-Moody algebra $U_{q}(\widehat{sl_2})$ is generated by the elements
$\{H_j,E_j,F_j\}$, $j\in \{0,1\}$. Denote the entries of the extended Cartan matrix\,\footnote{With $i,j\in\{0,1\}$: $a_{ii}=2$,\ $a_{ij}=-2$ for $i\neq j$.} as $\{a_{ij}\}$. The defining relations are:
\beqa [H_i,H_j]=0\ , \quad [H_i,E_j]=a_{ij}E_j\ , \quad
[H_i,F_j]=-a_{ij}F_j\ ,\quad
[E_i,F_j]=\delta_{ij}\frac{q^{H_i}-q^{-H_i}}{q-q^{-1}}\
\nonumber\eeqa
together with the $q-$Serre relations
\beqa [E_i,[E_i,[E_i,E_j]_{q}]_{q^{-1}}]=0\ ,\quad \mbox{and}\quad
[F_i,[F_i,[F_i,F_j]_{q}]_{q^{-1}}]=0\ . \label{defUq}\eeqa
The sum ${\it K}=H_0+H_1$ is the central element of the algebra.
\vspace{1mm}

We endow $U_{q}(\widehat{sl_2})$  with a
comultiplication $\Delta: U_{q}(\widehat{sl_2})\rightarrow U_{q}(\widehat{sl_2})\otimes U_{q}(\widehat{sl_2})$ 
 with
\beqa \Delta(E_i)&=&E_i\otimes q^{H_i/2} +
q^{-H_i/2}\otimes E_i\ ,\label{coprod} \\
 \Delta(F_i)&=&F_i\otimes q^{H_i/2} + q^{-H_i/2}\otimes F_i\ ,\nonumber\\
 \Delta(H_i)&=&H_i\otimes I\!\!I + I\!\!I \otimes H_i.\nonumber
\eeqa
%
%
%
More generally, one defines the $N-$coproduct $\Delta^{(N)}: \
U_{q}(\widehat{sl_2}) \longrightarrow
U_{q}(\widehat{sl_2}) \otimes \cdot\cdot\cdot \otimes
U_{q}(\widehat{sl_2})$ as
\beqa \Delta^{(N)}\equiv (id\times \cdot\cdot\cdot \times id
\times \Delta)\circ \Delta^{(N-1)}\ \label{coprodN}\eeqa
for $N\geq 3$ with $\Delta^{(2)}\equiv \Delta$,
$\Delta^{(1)}\equiv id$.
%
%
%
\vspace{1mm}

\subsection{The evaluation representation of  $U_{q}(\widehat{sl_2})$ (quantum loop algebra of $sl_2$) \cite{J2,CP}}\label{A2}
Infinite dimensional representations of  $U_{q}(\widehat{sl_2})$ associated with $K\equiv 0$ are the so-called  
 `evaluation representations'. They are constructed as follows. First, one introduces the evaluation homomorphism $\pi_v: U_{q}(\widehat{sl_2}) \mapsto U_{q}(sl_2)$ in the so-called principal gradation \cite{J2}:   
\beqa &&\pi_v[E_1]= vS_+\ , \qquad \ \ \ \pi_v[E_0]= vS_-\ , \label{homeval}\\
&&\pi_v[F_1]=
v^{-1}S_-\ ,\qquad \pi_v[F_0]= v^{-1}S_+\ ,\nonumber\\
\ &&\pi_v[q^{H_1/2}]= q^{s_3}\ ,\qquad \ \pi_v[q^{H_0/2}]=
q^{-s_3}\ ,\nonumber\eeqa
where  $v$ is called the evaluation parameter and the generators of $U_{q}(sl_2)$ satisfy
\beqa
[s_3,S_\pm]=\pm S_\pm \qquad \mbox{and} \qquad
[S_+,S_-]=\frac{q^{2s_3}-q^{-2s_3}}{q-q^{-1}}\ .\label{Uqsl2}
\eeqa
Note that the central element of $U_q(sl_2)$ is the Casimir operator:
\beqa
\Omega = \frac{q^{-1}q^{2s_3}+ q q^{-2s_3}}{(q-q^{-1})^2} + S_+S_-.\label{Casimir}
\eeqa

Let ${V}_j$ denote an irreducible finite dimensional representation of $U_q(sl_2)$. On $V_j$, the eigenvalue of $\Omega$ is given by:
\beqa
\omega_j=\frac{(q^{2j+1} + q^{-2j-1})}{(q-q^{-1})^2}.
\eeqa
Introduce the set of evaluation parameters $\{v_i|i=1,...,N\}$. An evaluation representation of $U_{q}(\widehat{sl_2})$ is given by
${\cal V}_{j}(v)\equiv {\mathbb C}[v,v^{-1}]\otimes V_j$, with the Chevalley generators  of $U_{q}(\widehat{sl_2})$ represented according to (\ref{homeval}) \cite{J2}. More generally, using the $N-$coproduct homomorphism (\ref{coprodN}), $N-$tensor products of evaluation representations denoted ${\cal V}^{(N)}$ are built as follows:
\beqa
{\cal V}^{(N)}= {\cal V}_{j_N}(v_N)\otimes  \cdots \otimes {\cal V}_{j_2}(v_2) \otimes {\cal V}_{j_1}(v_1).
\eeqa
Under certain conditions on the parameters $v_i,i=1,...,N$, ${\cal V}^{(N)}$  is irreducible \cite[Section 4.8]{CP}.
\vspace{1mm}

\subsection{The $q-$difference operators realization of $U_q(sl_2)$}\label{A3}
Irreducible finite dimensional representations $V_j$ of dimension $2j+1$ can be realized by $q-$difference operators acting in the linear space of one-variable polynomials $\mathcal{P}^{(1)}_z$ of degree $2j$.  There exists an homomorphism \cite{S1}:
\beqa
&& q^{s_3}\mapsto q^{-j}T_+, \qquad \qquad \qquad \qquad q^{-s_3}\mapsto q^{j}T_-,\\
&& S_+ \mapsto z\frac{(q^{2j}T_- - q^{-2j}T_+)}{(q-q^{-1})},\qquad S_-\mapsto -z^{-1}\frac{(T_- - T_+)}{(q-q^{-1})}.\nonumber
\eeqa

\vspace{3mm}

\section{The expansion coefficients of Example \ref{ex2}}
Note that all coefficients below depend implicitly on $n_1,j_1$, see (\ref{zpmi}) and (\ref{spec1}).

\beqa
A_2&=& \frac{z_-^{(2)}q^{2-4j_2}(1+q^2)(1-q^{-2n_2})\left({z_+^{(2)}}^2q^{4n_2-8j_2} - z_+^{(2)}z_-^{(2)}q^{-4j_2}(q^{2n_2}+1) + {z_-^{(2)}}^2q^{-2n_2}  \right)}{(z_-^{(2)}q^{-2n_2} -z_+^{(2)}q^{2-4j_2+2n_2})(z_-^{(2)}q^{2-2n_2} -z_+^{(2)}q^{2-4j_2+2n_2})(z_+^{(2)}q^{-4j_2+2n_2} -z_-^{(2)}q^{2-2n_2})}|_{n_1 \rightarrow n_1+1},\nonumber\\
B_2&=& \frac{q^{2-4j_2}(z_-^{(2)} -z_+^{(2)}q^{2-4j_2+2n_2})(z_-^{(2)} -z_+^{(2)}q^{2n_2-4j_2})}{(z_-^{(2)}q^{2-2n_2} -z_+^{(2)}q^{2-4j_2+2n_2})(z_-^{(2)}q^{-2n_2} -z_+^{(2)}q^{2-4j_2+2n_2})}|_{n_1 \rightarrow n_1+1},\nonumber\\
C_2&=& \frac{{z_-^{(2)}}^2q^{2-4j_2}(q^{2-2n_2}-1)(q^{-2n_2}-1)}{(z_-^{(2)}q^{2-2n_2} -z_+^{(2)}q^{2-4j_2+2n_2})(z_-^{(2)}q^{2-2n_2} -z_+^{(2)}q^{-4j_2+2n_2})}|_{n_1 \rightarrow n_1+1},\nonumber
\eeqa
\beqa
A'_2&=& \frac{z_+^{(2)}q^{2-4j_2}(1+q^2)
(q^{-4j_2}-q^{-2n_2})\left({z_+^{(2)}}^2q^{4n_2-4j_2} - z_+^{(2)}z_-^{(2)}q^{-4j_2}(q^{4j_2}+q^{2n_2}) + {z_-^{(2)}}^2q^{-2n_2}  \right)}{(z_+^{(2)}q^{2n_2-4j_2} -z_-^{(2)}q^{2-2n_2})(z_+^{(2)}q^{2-4j_2+2n_2} -z_-^{(2)}q^{2-2n_2})(z_+^{(2)}q^{2-4j_2+2n_2} -z_-^{(2)}q^{-2n_2})}|_{n_1 \rightarrow n_1-1}
,\nonumber\\
B'_2&=& \frac{{z_+^{(2)}}^2q^{2-4j_2}(q^{2-4j_2+2n_2}-1)(q^{2n_2-4j_2}-1)}{(z_+^{(2)}q^{2-4j_2+2n_2} -z_-^{(2)}q^{-2n_2})(z_+^{(2)}q^{2-4j_2+2n_2} -z_-^{(2)}q^{2-2n_2})}|_{n_1 \rightarrow n_1-1},\nonumber\\
C'_2&=& \frac{q^{2-4j_2}(z_-^{(2)}q^{-2n_2} -z_+^{(2)})(z_-^{(2)}q^{2-2n_2} -z_+^{(2)})}{(z_+^{(2)}q^{2-4j_2+2n_2} -z_-^{(2)}q^{2-2n_2})(z_+^{(2)}q^{2n_2-4j_2} -z_-^{(2)}q^{2-2n_2})}|_{n_1 \rightarrow n_1-1}\nonumber
\eeqa
and\footnote{$[n]_q=\frac{q^n-q^{-n}}{q-q^{-1}}$.}
\beqa
A_{aux}&=& \frac{ (q^{-2j_2}z_+^{(2)}z_-^{(2)} b_1^{(2)} + + q^{-2}c_1^{(2)})u^{(2)} + q^{-2}{\cal A}_{n_1}^{[0]} x^{(2)} }{(q^{4n_2}z_+^{(2)} - q^{4j_2}z_-^{(2)})
(q^{4n_2} z_+^{(2)}   - q^{4j_2+2}z_-^{(2)}) (q^{4n_2} z_+^{(2)}   - q^{4j_2-2}z_-^{(2)})}
,\nonumber \\
B_{aux}&=&\frac{(1-q^{2n_2-4j_2})q^{4n_2+8j_2} \left(  q^{4n_2-6j_2+4}{z_+^{(2)}}^2b_1^{(2)}  + q^{2n_2-2j_2+2}z_+^{(2)}{\cal A}_{n_1}^{[0]} +c_1^{(2)}\right)v^{(2)}}{(q^{4n_2}z_+^{(2)} - q^{4j_2}z_-^{(2)})(q^{4n_2} z_+^{(2)}   - q^{4j_2+2}z_-^{(2)}) (q^{4n_2} z_+^{(2)}   - q^{4j_2-2}z_-^{(2)})} ,\nonumber \\
C_{aux}&=& \frac{(1-q^{2n_2})q^{4n_2+8j_2-2} \left(  q^{-4n_2+2j_2+4}{z_-^{(2)}}^2b_1^{(2)}  + q^{-2n_2+2j_2+2}z_-^{(2)}{\cal A}_{n_1}^{[0]} +c_1^{(2)}\right)w^{(2)}}{(q^{4n_2}z_+^{(2)} - q^{4j_2}z_-^{(2)})(q^{4n_2} z_+^{(2)}   - q^{4j_2+2}z_-^{(2)}) (q^{4n_2} z_+^{(2)}   - q^{4j_2-2}z_-^{(2)})} \nonumber
\eeqa
where
\beqa
x^{(2)}&=& q^{4n_2+6j_2}\left((1+q^2) (q^{4j_2+2}+1)-q^{-2n_2+4j_2+2}\right)  z_+^{(2)}{z_-^{(2)}}^2\nonumber\\
&&\quad - \  q^{8n_2+2j_2}\left( (1+q^2) (q^{4j_2+2}+1)-q^{2n_2+2}\right)  {z_+^{(2)}}^2{z_-^{(2)}}\nonumber\\
&&\quad + \ q^{10n_2+2j_2+2}{z_+^{(2)}}^3
- \ q^{2n_2+10j_2+2}{z_-^{(2)}}^3,
\nonumber\\
u^{(2)}&=&q^{6n_2+8j_2}(q^{2n_2-4j_2}-q^{-2n_2})(1+q^2)z_+^{(2)}z_-^{(2)}\nonumber\\
&& \quad + \ \left(q^{8n_2+4j_2}(1+q^2) -q^{10n_2}(1+q^{4j_2+2})\right){z_+^{(2)}}^2\nonumber\\
&& \quad - \ \left(q^{4n_2+8j_2}(1+q^2) -q^{2n_2+8j_2}(1+q^{4j_2+2})\right){z_-^{(2)}}^2,\nonumber\\
v^{(2)}&=&(1+q^{2n_2-2})z_+^{(2)}z_-^{(2)} -q^{4j_2-2n_2}{z_-^{(2)}}^2 -q^{4n_2-4j_2-2}z_+^{(2)},\nonumber\\
w^{(2)}&=&(1+q^{2n_2-4j_2+2})z_+^{(2)}z_-^{(2)} -q^{-2n_2}{z_-^{(2)}}^2 -q^{4n_2-4j_2+2}z_+^{(2)}.\nonumber
\eeqa

\vspace{3mm}

\section{The $q-$Racah and the dual $q-$Krawtchouk polynomials}
\subsection{The $q-$Racah polynomials:} Let $M$ be a positive integer and $n=0,1,...,M$. We denote the $q-$Racah polynomial  $R_n(\mu(\tilde{n}))$ with argument $\mu(\tilde{n})=q^{-2\tilde{n}} + \gamma\delta q^{2(\tilde{n}+1)}$ as \cite[Section 3.2]{KS}:
\beqa
\qquad \quad R_n(\mu(\tilde{n});\alpha,\beta,\gamma,\delta;q^2) &=&\fpt{q^{-2n}} { \alpha\beta q^{2(n+1)} }{q^{-2\tilde{n}}}{\gamma\delta q^{2(\tilde{n}+1)}}{\alpha q^2}{\beta\delta q^2}{\gamma q^2} \quad \mbox{with} \quad \gamma=q^{-2M-2}.\label{qR}
\eeqa
The $q-$Racah polynomials satisfy the orthogonality condition:
\beqa
\sum_{\tilde{n}=0}^{M} \frac{(\gamma\delta q^2,\alpha q^2,\beta\delta q^2,\gamma q^2;q^2)_{\tilde{n}}}{(q,\gamma\delta q^2/\alpha,\gamma q^2/\beta, \delta q^2;q^2)_{\tilde{n}}}\frac{(1-\gamma\delta q^{4\tilde{n}+2})}{(\alpha\beta q^2)^{\tilde{n}}(1-\gamma\delta q^2)}  R_m(\mu(\tilde{n})) R_n(\mu(\tilde{n}))=h_n\delta_{mn}\label{orthoqRac}
\eeqa
where
\beqa
h_n=\frac{(\alpha\beta q^4,1/\delta;q^2)_N}{(\beta q^2,\alpha q^2/\delta;q^2)_N}\frac{(1-\alpha\beta q^2)(\delta q^{-2M})^n}{(1-\alpha\beta q^{4n+2})}\frac{(q^2,\beta q^2,\alpha q^2/\delta,\alpha\beta q^{2M+4};q^2)_n}{(\alpha\beta q^2,\alpha q^2,\beta\delta q^2,q^{-2M};q^2)_n}\nonumber
\eeqa
and the three-term recurrence relations: 
\beqa
\mu(\tilde{n})R_n(\mu(\tilde{n})) = B_n R_{n+1}(\mu(\tilde{n})) +(1+\gamma\delta q^{2}-B_n - C_n)R_{n}(\mu(\tilde{n})) +  C_n R_{n-1}(\mu(\tilde{n}))\label{threerec1}
\eeqa
where
\beqa
B_n &\equiv& B_n(\alpha,\beta,\gamma,\delta) =\frac{(1-\alpha q^{2n+2})(1-\gamma q^{2n+2})(1-\alpha\beta q^{2n+2})(1-\beta\delta q^{2n+2})}{(1-\alpha\beta q^{4n+2})(1-\alpha\beta q^{4n+4})},\label{Bcoeff}\\
C_{n}&\equiv&C_n(\alpha,\beta,\gamma,\delta)=\frac{q^2(1- q^{2n})(1-\beta q^{2n})(\delta-\alpha q^{2n})(\gamma-\alpha\beta q^{2n})}{(1-\alpha\beta q^{4n})(1-\alpha\beta q^{4n+2})}.\label{Ccoeff}
\eeqa

\vspace{1mm}

\subsection{The dual $q-$Krawtchouk polynomials:}\label{B3}  Let $M$ be a positive integer and $n=0,1,...,M$. We denote the dual $q-$Krawtchouk polynomial  $K_n(\mu(\tilde{n}))$ with argument $\mu(\tilde{n})=q^{-2\tilde{n}} + c q^{2(\tilde{n}-M)}$ as \cite[Section 3.17]{KS}:
\beqa
\qquad \quad K_n(\mu(\tilde{n});c,M;q^2) &=&\fpts{q^{-2n}}{q^{-2\tilde{n}}} {c q^{2(\tilde{n}-M)}}{ q^{-2M}}{ q^{-2M}}{ 0} .\label{qK}
\eeqa
The dual $q-$Krawtchouk polynomials satisfy the orthogonality condition for $c< 0$:
\beqa
&& \sum_{\tilde{n}=0}^{M} \frac{( cq^{-2M},q^{-2M};q^2)_{\tilde{n}}}{(q^2,cq^2;q^2)_{\tilde{n}}}
\frac{(1-c q^{4\tilde{n} -2M})}{(1-c q^{-2M})}  c^{-\tilde{n}} q^{2\tilde{n}(2M-\tilde{n})}K_m(\mu(\tilde{n})) K_n(\mu(\tilde{n}))= \label{orthoqK} \\
&& \qquad\qquad     \qquad\qquad \qquad\qquad \qquad\qquad \qquad\qquad \qquad\qquad \frac{(c^{-1};q^2)_M(q^2;q^2)_n}{(q^{-2M};q^2)_n}(cq^{-2M})^n\delta_{mn}\nonumber.
\eeqa
and the three-term recurrence relations: 
\beqa
\mu(\tilde{n})K_n(\mu(\tilde{n})) = (1-q^{2n-2M})K_{n+1}(\mu(\tilde{n})) + (1+c)q^{2n-2M}K_n(\mu(\tilde{n}))+  cq^{-2M}(1-q^{2n})K_{n-1}(\mu(\tilde{n}))\label{threerec2}\nonumber.
\eeqa
\vspace{1mm}

\subsection{Inversion and fusion formulas:}
The generating function of the dual $q-$Krawtchouk polynomials can be found in  \cite[eq. (3.17.11)]{KS}. Explicitly, one has:
\beqa
(cy;q^2)_{\tilde{n}}(dy;q^2)_{M-\tilde{n}}=\sum_{n=0}^{M} (dq^{2M})^n \frac{(q^{-2M};q^2)_n}{(q^2;q^2)_n} K_n(\mu(\tilde{n});c/d,M;q^2) y^n
\eeqa
The relation can be inversed using the orthogonality condition (\ref{orthoqK}). It yields to:
\beqa
y^n=\sum_{l=0}^{M} \frac{(d/c)^l c^{-n}q^{2l(2M-l)}}{(d/c;q^{2})_{M}} \frac{ (q^{-2M}c/d,q^{-2M};q^{2})_{l}} {(q^{2},q^{2}c/d;q^{2})_{l}} \frac{(1-c q^{-2M+4l}/d)}{(1-c q^{-2M}/d)} K_n(\mu(l);c/d,M;q^{2})  (cy;q^2)_{l}(dy;q^2)_{M-l}. \nonumber
\eeqa
From this relation, it follows that the coefficients entering in (\ref{expC}) can be written in terms of products of dual $q-$Krawtchouk polynomials. Namely:
\beqa 
&&  \qquad  \qquad C_{n}^{\tilde{n}}(a,b,c,d;N;q^2)  = 
\frac{q^{2{\tilde{n}}(2N-{\tilde{n}})}}{(c/d)^{\tilde{n}}(d/c;q^{2})_{N}} \frac{ (q^{-2N}c/d,q^{-2N};q^{2})_{{\tilde{n}}}} {(q^{2},q^{2}c/d;q^{2})_{{\tilde{n}}}} \frac{(1-c q^{-2N+4{\tilde{n}}}/d)}{(1-c q^{-2N}/d)}\label{Cfuse}\\
&& \qquad  \qquad  \qquad  \qquad \qquad   \qquad \times \
\sum_{m=0}^{N}  \frac{q^{-2mN}}{ (c/b)^{m}}\frac{(q^{2N};q^{-2})_m}{(q^{-2};q^{-2})_m}
K_m(\mu(n)|_{q\rightarrow q^{-1}};a/b,N;q^{-2}) K_m(\mu(\tilde{n});c/d,N;q^{2}) .\nonumber
\eeqa

\vspace{3mm}

\end{appendix}

\end{document}